\documentclass[letterpaper,11pt]{article}
\usepackage[margin=1in]{geometry}
\usepackage{setspace}
\usepackage{graphicx}
\usepackage{epsfig}

\usepackage{microtype} 
\usepackage{array} 

\usepackage{amssymb}
\usepackage{amsmath}
\usepackage{amsthm}
\usepackage{mathtools} 
\usepackage{stmaryrd} 
\usepackage{stackrel} 
\usepackage{latexsym} 

\usepackage{tikz}
\usetikzlibrary{arrows}
\usetikzlibrary{arrows.meta}
\usetikzlibrary{decorations.pathreplacing}
\usetikzlibrary{cd}
\usepackage[all]{xy} 

\usepackage{algorithm}
\usepackage[noend]{algpseudocode}
\algnewcommand{\LineComment}[1]{\Statex\hspace{\algorithmicindent}\(\triangleright\) #1}
\algnewcommand\algorithmicforeach{\textbf{for each}}
\algdef{S}[FOR]{ForEach}[1]{\algorithmicforeach\ #1\ \algorithmicdo}
\usepackage{etoolbox} 

\usepackage{caption}
\usepackage{subcaption}

\usepackage{makecell} 
\usepackage{booktabs} 
\usepackage{multirow} 

\usepackage[hidelinks]{hyperref} 
\usepackage[symbol]{footmisc}
\usepackage{makeidx} 
\usepackage{url} 
\usepackage{color}
\usepackage{xcolor}
\usepackage{xspace} 
\usepackage[normalem]{ulem} 
\usepackage{soul} 
\usepackage{ifpdf} 
\usepackage{extarrows}

\makeatletter
\expandafter\patchcmd\csname\string\algorithmic\endcsname{\itemsep\z@}{\itemsep=0.25ex}{}{}
\makeatother

\newcounter{usesmallsep}
\ifnum\the\value{usesmallsep}=1

    \newlength{\myitemsep}
    \newlength{\mytopsep}
    \usepackage{enumitem}
    \setlist[itemize]{leftmargin=\parindent,parsep=\parskip,
      listparindent=\parindent,itemsep=\myitemsep,topsep=\myitemsep}
    \setlist[enumerate]{leftmargin=\parindent,parsep=\parskip,
      listparindent=\parindent,itemsep=\myitemsep,,topsep=\myitemsep}
    \setlist[description]{font=\bfseries,leftmargin=\parindent,parsep=\parskip,
      listparindent=\parindent,itemsep=\myitemsep,topsep=\myitemsep}
    
    \newlength{\mypartitlesep}
    \usepackage[explicit]{titlesec}
    \titlespacing{\paragraph}{0pt}{\mypartitlesep}{\mypartitlesep}
    
    \newlength{\mythmsep}
    \newtheoremstyle{mythmstyle}
      {\mythmsep} 
      {\mythmsep} 
      {\itshape} 
      {} 
      {\bfseries} 
      {.} 
      {.5em} 
      {} 
    
    \newtheoremstyle{mydefstyle}
      {\mythmsep} 
      {\mythmsep} 
      {} 
      {} 
      {\bfseries} 
      {.} 
      {.5em} 
      {} 
    
    \theoremstyle{mythmstyle}
        \newtheorem{theorem}{Theorem}
        \newtheorem{proposition}[theorem]{Proposition}
        \newtheorem{lemma}[theorem]{Lemma}
        \newtheorem{corollary}[theorem]{Corollary}
        \newtheorem{fact}[theorem]{Fact}
        \newtheorem*{fact*}{Fact}
    \theoremstyle{mydefstyle}
        \newtheorem{definition}{Definition}
        \newtheorem{problem}{Problem}
        \newtheorem{assumption}{Assumption}
        \newtheorem{remark}{Remark}
        \newtheorem{algr}[algorithm]{Algorithm}
    
    \newenvironment{proof}
        {\vspace{-0.9em}\begin{proof}}
        {\end{proof}\vspace{-0.4em}}

\else

    \theoremstyle{plain}
        \newtheorem{theorem}{Theorem}
        \newtheorem{proposition}[theorem]{Proposition}

        \newtheorem{fact}[theorem]{Fact}

        \newtheorem*{algr*}{Algorithm}
    \theoremstyle{definition}
        \newtheorem{definition}[theorem]{Definition}

        \newtheorem{remark}[theorem]{Remark}
        
    \usepackage{enumitem}
    \setlist[itemize]{leftmargin=\parindent}
    \setlist[enumerate]{leftmargin=\parindent}
    \setlist[description]{font=\bfseries,leftmargin=\parindent}

\fi

\newcommand{\Hm}{\mathsf{H}}
\newcommand{\Hmr}{\mathsf{\tilde{H}}}

\newcommand{\Hom}{\mathsf{Hom}}

\newcommand{\Real}{\mathbb{R}}

\newcommand{\fsimp}[2]{\sigma_{#2}}

\newcommand{\id}{\mathtt{id}}

\newcommand{\Pers}{\mathsf{Pers}}

\newcommand{\inv}{^{-1}}

\newcommand{\lbarrowspace}{\;}

\let\leftrightarrowsp\lrarrowsp

\newcommand{\incto}{\hookrightarrow}
\newcommand{\inctosp}[1]{\xhookrightarrow{\lbarrowspace#1\lbarrowspace}}
\newcommand{\bakincto}{\hookleftarrow}
\newcommand{\bakinctosp}[1]{\xhookleftarrow{\lbarrowspace#1\lbarrowspace}}

\newcommand{\given}{\,|\,}
\newcommand{\Set}[1]{\{#1\}}

\let\emptyset\varnothing

\let\intsec\intersect
\let\union\cup
\let\bigunion\bigcup

\newcommand{\Ecal}{\mathcal{E}}
\newcommand{\Dcal}{\mathcal{D}}
\newcommand{\Fcal}{\mathcal{F}}

\newcommand{\Ical}{\mathcal{I}}

\newcommand{\Lcal}{\mathcal{L}}
\newcommand{\Mcal}{\mathcal{M}}

\newcommand{\Ucal}{\mathcal{U}}

\newcommand{\Zbb}{\mathbb{Z}}

\newcommand{\aG}{\alpha}
\newcommand{\bG}{\beta}

\newcommand{\lG}{\lambda}
\newcommand{\LG}{\Lambda}

\newcommand{\oG}{\omega}

\newcommand{\sG}{\sigma}

\newcommand{\tG}{\tau}

\newcommand{\Dim}{p}

\newcommand{\birth}{b}
\newcommand{\death}{d}
\newcommand{\filtcnt}{m}
\newcommand{\simpcnt}{n}

\newcommand{\splx}{\sigma}
\newcommand{\ud}{\Ucal}
\newcommand{\ucplx}{L}
\newcommand{\usimp}{\tau}
\newcommand{\ef}{\Ecal}
\newcommand{\add}[1]{\mathsf{a}#1}
\newcommand{\del}[1]{\mathsf{d}#1}

\newcommand{\dfilt}{\Dcal}

\begin{document}

\title{On Association between Absolute and Relative Zigzag Persistence\thanks{This research is supported by NSF grants CCF 1839252 and 2049010.}}

\author{Tamal K. Dey\hspace{8em}Tao Hou\vspace{1em}\\
{\footnotesize 
Department of Computer Science, Purdue University. 
\texttt{tamaldey,hou145@purdue.edu}
}
}

\date{}

\maketitle
\thispagestyle{empty}

\begin{abstract}
Duality results~\cite{cohen2009extending,de2011dualities}
connecting persistence modules for
absolute and relative homology provides a fundamental
understanding into persistence theory.
In this paper,
we study similar associations in the context of
zigzag persistence.
Our main finding is
a \emph{weak duality} for the so-called \emph{non-repetitive} zigzag filtrations
in which a simplex is never added again after being deleted.
The technique used to prove the duality for
non-zigzag persistence does not extend straightforwardly
to our case. 
Accordingly, taking a different route,
we prove the weak duality 
by converting a non-repetitive filtration to
an \emph{up-down} filtration by a sequence of \emph{diamond switches}~\cite{carlsson2010zigzag}. 
We then show an application of 
the weak duality result which gives a near-linear algorithm for computing 
the $p$-th and a subset of the $(p-1)$-th persistence for
a non-repetitive zigzag filtration of a simplicial $p$-manifold.
Utilizing the fact 
that a non-repetitive filtration admits an up-down filtration
as its canonical form, we further reduce the problem of
computing zigzag persistence for non-repetitive filtrations to the
problem of computing standard persistence for which several
efficient implementations exist.
Our experiment shows that
this achieves substantial performance gain.
Our study also identifies repetitive filtrations as instances that
fundamentally distinguish zigzag persistence from the standard persistence.
\end{abstract}

\newpage
\setcounter{page}{1}


\section{Introduction}
Standard persistent homology defined over a growing sequence 
of simplicial complexes is
a fundamental tool in topological data analysis (TDA). Since the advent of
persistence algorithm~\cite{edelsbrunner2000topological} and its algebraic
understanding~\cite{zomorodian2005computing}, various extensions of the basic concept
have been explored~\cite{carlsson2010zigzag,carlsson2009zigzag-realvalue,cohen2009extending,de2011dualities}.
These extensions
create a need for finding relations among them for a more coherent theory,
which can potentially lead to improved algorithms.
Two excellent examples catering to this need are
the papers by de Silva et al.~\cite{de2011dualities} and Cohen-Steiner et al.~\cite{cohen2009extending},
which relate standard persistence modules
for \emph{absolute} and \emph{relative}
homology. The original persistence algorithm~\cite{edelsbrunner2000topological} was designed for absolute homology,
but in practice persistence modules involving relative homology do appear~\cite{cohen2009extending,dey2020persistence}. Although
the persistence for relative homology can be computed 
by `coning' on the complexes,
the authors of~\cite{de2011dualities} show a duality between the modules arising
from absolute and relative homology, 
thereby revealing an equivalence in their computations. 

In this paper, we explore a similar correspondence
between the two types of modules but in the context of {\it zigzag} persistence.
Zigzag persistence introduced by Carlsson and de Silva~\cite{carlsson2010zigzag}
empowered TDA to deal with 
filtrations where both insertion
and deletion of simplices are allowed. 
In practice, allowing deletion of simplices does make the topological tool more powerful.
For example, in dynamic networks~\cite{holme2012temporal},
 a sequence of graphs may not grow monotonically but can also shrink due to
disappearance of vertex connections. 

Zigzag persistence possesses some key differences
from standard persistence. For example,
unlike standard (non-zigzag) modules which decompose
into only finite and infinite intervals, zigzag modules
decompose into {\it four types} of intervals (see Definition~\ref{dfn:open-close-bd}). 
This motivates us to have a fresh look at 
relating zigzag modules arising from absolute and relative homology,
and it 
turns out that the technique 
used in the non-zigzag case~\cite{de2011dualities} does not extend straightforwardly. 
Using
the Mayer-Vietoris diamond proposed by Carlsson and de Silva~\cite{carlsson2010zigzag}, we arrive
at a duality result which is {\it weak} in the sense that
an interval from the absolute zigzag module
may correspond to two intervals from the relative zigzag module (see Theorem~\ref{thm:dual}). Furthermore, this (weak) duality exists only for
the {\it non-repetitive} 
filtrations where a simplex is never inserted again
after its deletion (see Definition~\ref{dfn:nonrep-filt}). 


We demonstrate an application of the weak duality result
by proposing an efficient algorithm for computing zigzag persistence
for non-repetitive filtrations of simplicial $\Dim$-manifolds.
The algorithm utilizes
the Lefschetz duality~\cite{cohen2009extending,munkres2018elements} 
and a recent 
near-linear algorithm for zigzag persistence on graphs~\cite{dey2021graph}.
It computes all $\Dim$-th intervals
and
all $(\Dim-1)$-th intervals (except for one type)
in near-linear time. This improves the current best time bound of $O(n^{\omega})$ incurred by applying the algorithm for computing general zigzag persistence~\cite{milosavljevic2011zigzag} to this special case.

One insight we gain from proving the weak duality is
the fact that a non-repetitive zigzag filtration
admits an {\it up-down} filtration as its canonical form.
We further discover that the up-down filtration can be converted
into a {\it non-zigzag} filtration 
and the barcode of the resulting filtration can be related to the original one. 
This leads to
an efficient algorithm for 
computing zigzag persistence
for any non-repetitive filtrations.
Note that
algorithms
for zigzag persistence~\cite{carlsson2009zigzag-realvalue,maria2014zigzag} are 
more involved (and hence slower in practice)
than algorithms for
the non-zigzag version though they have the same time complexity~\cite{milosavljevic2011zigzag}.
By leveraging our results, we are able to compute zigzag persistence for
non-repetitive filtrations using any 
efficient software for computing standard persistence.
In fact, our experiments show that the
gain in practice is substantial. However, the reader should be aware of
the caveat that the input
filtration has to be non-repetitive. The importance of this condition has not been pointed
out in the literature before. In a sense, this condition  
is satisfied by
a milder version of zigzag persistence called {\it levelset zigzag}~\cite{carlsson2009zigzag-realvalue}.
Last but not the least, our finding helps to pinpoint
the repetitive filtrations as the instances of zigzag filtrations
that fundamentally distinguish zigzag persistence
from the standard persistence.


\section{Preliminaries}

\paragraph{Absolute and relative homology.}
We briefly mention some of the algebraic structures adopted in this paper;
see~\cite{hatcher2002algebraic,munkres2018elements} for details.
All homology groups are taken with coefficient $\Zbb_2$
and therefore vector spaces mentioned in this paper are also over $\Zbb_2$.
Let $K$ be a simplicial complex.
For $\Dim\geq 0$, $\Hm_{\Dim}(K)$ denotes the {\it $\Dim$-th homology group} of $K$.
We also let $\Hm_{*}(K)$ denote the homology group of all dimensions, 
i.e., $\Hm_{*}(K)=\bigoplus_{\Dim\geq 0}\Hm_{\Dim}(K)$.
Relative homology groups are frequently used in this paper. 
Specifically, given a {\it relative pair of simplicial complexes}
$(K,L)$ where $L\subseteq K$,
we denote the {\it $\Dim$-th relative homology group} of $(K,L)$
as $\Hm_{\Dim}(K,L)$
and denote the relative homology group of all dimensions as $\Hm_{*}(K,L)$.
Sometimes to differentiate, 
we also call $\Hm_{\Dim}(K)$ the $\Dim$-th {\it absolute} homology group of $K$.

\paragraph{Zigzag module, barcode, and filtration.}
A {\it zigzag module}~\cite{carlsson2010zigzag} (or simply {\it module})
is a sequence of vector spaces 
\[\Mcal: V_0 
\leftrightarrow
V_1 
\leftrightarrow
\cdots
\leftrightarrow
V_\filtcnt\]
in which
each $V_i\leftrightarrow V_{i+1}$ is a linear map and is either forward, i.e., $V_i\to V_{i+1}$,
or backward, i.e., $V_i\leftarrow V_{i+1}$.
It is known~\cite{carlsson2010zigzag,Gabriel72} that
$\Mcal$ has a decomposition of the form
$\Mcal\simeq\bigoplus_{k\in\LG}\Ical^{[\birth_k,\death_k]}$,
in which each $\Ical^{[\birth_k,\death_k]}$
is a special type of module called {\it interval module} over the interval $[\birth_k,\death_k]$.
The (multi-)set of intervals
denoted as
$\Pers(\Mcal):=\Set{[\birth_k,\death_k]\given k\in\LG}$
is an invariant of $\Mcal$
and is called the {\it zigzag barcode} (or simply {\it barcode}) of $\Mcal$.
Each interval in a zigzag barcode is called a {\it persistence interval}.
The following definition characterizes different types of persistence intervals:
\begin{definition}[Open and closed birth/death]
\label{dfn:open-close-bd}
Let $\Mcal:V_0\leftrightarrow V_1\leftrightarrow\cdots\leftrightarrow V_{\filtcnt}$ 
be a zigzag module.
We call
the start of any interval in $\Pers(\Mcal)$ as a {\it birth index}
in $\Mcal$
and call the end of any interval a {\it death index}.
Moreover, 
we call a birth index $\birth$ as {\it closed} 
if $\birth=0$, or $\birth>0$ and $V_{\birth-1}\rightarrow V_\birth$ is a forward map;
otherwise, we call $\birth$ {\it open}.
Symmetrically, 
we call
a death index $\death$ as {\it closed}
if $\death=\filtcnt$, or $\death<\filtcnt$ and $V_{\death}\leftarrow V_{\death+1}$ is a backward map;
otherwise, we call $\death$ {\it open}.
The types of the birth/death ends
classify intervals in $\Pers(\Mcal)$ into four types: 
\emph{closed-closed}, \emph{closed-open}, \emph{open-closed}, and \emph{open-open}. 
\end{definition}
\begin{remark}
In this paper, we always denote a persistence interval
as an interval of integers which is of the form $[\birth,\death]$.
Hence, other than the cases when $\birth=0$ or $\death=\filtcnt$,
the designation of open and closed ends of $[\birth,\death]$
is determined by the directions of the maps 
$V_{\birth-1}\leftrightarrow V_\birth$,
$V_{\death}\leftrightarrow V_{\death+1}$ in $\Mcal$. Moreover,
if $\Mcal$ is the module for the levelset zigzag~\cite{carlsson2009zigzag-realvalue} of a function,
then the open and closed ends defined above are the same as the open and closed ends
for levelset zigzag.
\end{remark}

A {\it zigzag filtration} (or simply {\it filtration})
is a sequence of simplicial complexes 
\[\Fcal: K_0 \leftrightarrow K_1 \leftrightarrow 
\cdots \leftrightarrow K_\filtcnt\]
in which each
$K_i\leftrightarrow K_{i+1}$ is either a forward inclusion $K_i\incto K_{i+1}$
or a backward inclusion $K_i\bakincto K_{i+1}$.
Note that a forward (resp.\ backward) inclusion
in $\Fcal$ can be considered as an addition (resp.\ deletion)
of zero, one, or more simplices.
Moreover,
we call
$\Fcal$ an {\it up-down} filtration~\cite{carlsson2009zigzag-realvalue} if 
$\Fcal$ can be separated into two parts such that
the first part contains only forward inclusions
and the second part contains only backward ones,
i.e., $\Fcal$ is of the form $\Fcal: K_0 \incto K_1 \incto 
\cdots \incto K_{\ell} \bakincto K_{\ell+1} \bakincto \cdots \bakincto K_\filtcnt$.

The $\Dim$-th homology groups
induce the {\it $\Dim$-th zigzag module} of $\Fcal$
\[\Hm_\Dim(\Fcal): 
\Hm_\Dim(K_0) 
\leftrightarrow
\Hm_\Dim(K_1) 
\leftrightarrow
\cdots 
\leftrightarrow
\Hm_\Dim(K_\filtcnt) \]
in which
each $\Hm_\Dim(K_i)\leftrightarrow \Hm_\Dim(K_{i+1})$
is a linear map induced by inclusion.
The barcode $\Pers(\Hm_\Dim(\Fcal))$ is also called the {\it$\Dim$-th zigzag barcode}
of $\Fcal$ and is alternatively denoted as
$\Pers_\Dim(\Fcal):=\Pers(\Hm_\Dim(\Fcal))$.
Each persistence interval in $\Pers_\Dim(\Fcal)$
is said to have dimension $\Dim$.
Frequently in this paper, we consider the homology in all dimensions
and take the zigzag module $\Hm_*(\Fcal)$,
for which we have $\Pers_*(\Fcal)=\bigsqcup_{\Dim\geq 0}\Pers_\Dim(\Fcal)$.

Letting $K=\bigunion_{i=0}^\filtcnt K_i$,
we call $K$ the {\it total complex} of $\Fcal$ and
call $\Fcal$ a filtration {\it of} $K$.
Using the total complex,
we define the {\it relative filtration} $K/\Fcal$ of $\Fcal$ as
\[K/\Fcal: (K,K_0) \leftrightarrow (K,K_1) \leftrightarrow 
\cdots \leftrightarrow (K,K_\filtcnt).\]
Then, zigzag modules $\Hm_\Dim(K/\Fcal)$ and $\Hm_*(K/\Fcal)$
can be induced using relative homology.
We also alternatively denote the barcodes of the induced modules as 
$\Pers_\Dim(K/\Fcal)$ and $\Pers_*(K/\Fcal)$.
We sometimes call $\Hm_\Dim(K/\Fcal)$, $\Hm_*(K/\Fcal)$ as {\it relative} modules
and $\Pers_\Dim(K/\Fcal)$, $\Pers_*(K/\Fcal)$
as {\it relative} barcodes.
Similarly, 
we also call $\Fcal$, $\Hm_\Dim(\Fcal)$, $\Hm_*(\Fcal)$,
$\Pers_\Dim(\Fcal)$, and $\Pers_*(\Fcal)$
as {\it absolute} filtration, modules, and barcodes respectively.

One type of (absolute) filtration
called {\it simplex-wise} filtration
is especially useful,
in which 
each forward (resp.\ backward) inclusion 
is an addition (resp.\ deletion) of a single simplex.
Since each zigzag filtration can be made simplex-wise
by expanding each inclusion into a series of simplex-wise inclusions, 
we do not lose generality by considering only simplex-wise filtrations.
We sometimes denote an inclusion in a simplex-wise filtration
as $K_0\inctosp{\sG} K_1$ (for forward case)
or $K_0\bakinctosp{\sG} K_1$ (for backward case),
so that the simplex being added or deleted is clear.
Hence, a simplex-wise filtration $\Fcal$ can be of the form
$\Fcal:
K_0\leftrightarrowsp{\fsimp{}{0}} K_1\leftrightarrowsp{\fsimp{}{1}}
\cdots 
\leftrightarrowsp{\fsimp{}{\filtcnt-1}} K_\filtcnt
$.

\begin{remark}\label{rmk:inc-bd}
An inclusion $K_i\leftrightarrow K_{i+1}$ in a simplex-wise filtration 
either provides $i+1$ as a birth index
or provides $i$ as a death index (but cannot provide both).
\end{remark}

We then define the following type of (absolute) filtration which we focus on:
\begin{definition}[Non-repetitive filtration]
\label{dfn:nonrep-filt}
A zigzag filtration is said to be \emph{non-repetitive} if 
whenever a simplex $\sG$ is deleted from the filtration,
the simplex $\sG$ is never added again.
\end{definition}
\begin{remark}
A well-known type of non-repetitive filtrations
are the (discretized) filtrations for levelset zigzag persistence~\cite{carlsson2009zigzag-realvalue},
which is termed as {\it levelset filtration} in this paper.
\end{remark}



\section{Duality of absolute and relative zigzag}
\label{sec:duality}

In this section, we 
present and prove a weak duality theorem 
between absolute and relative zigzag persistence. This duality
requires that the filtration generating the persistence be non-repetitive.
To establish the duality result 
for a non-repetitive, simplex-wise filtration 
$\Fcal:K_0\leftrightarrow
K_1\leftrightarrow
\cdots
\leftrightarrow K_\filtcnt$,
we first turn $\Fcal$ into a more standard form
as follows:
we attach additions to the beginning of $\Fcal$
and deletions to the end, 
to derive another simplex-wise filtration $\Fcal'$ which
starts and ends with empty complexes.
Note that $\Fcal'$ is also non-repetitive.
Since the absolute and relative barcodes of $\Fcal$
can be easily derived from those of $\Fcal'$,
our duality result focuses on this standard form only:
\def\mygraphic{\includegraphics[height=5.5pt]{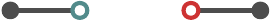}}
\newlength\ccrelwidth
\setlength\ccrelwidth{\widthof{\mygraphic}}
\begin{theorem}[Weak duality of absolute/relative zigzag]
\label{thm:dual}
Let $K$ be a simplicial complex and 
$
\Fcal:\emptyset=K_0\leftrightarrow K_1\leftrightarrow
\cdots
\leftrightarrow K_\filtcnt=\emptyset
$
be a simplex-wise filtration of $K$
which is non-repetitive.
Then,
there exists the following surjective correspondence
from intervals in $\Pers_*(\Fcal)$ to intervals in $\Pers_*(K/\Fcal)${\rm:}
\begin{center}
{\rm\begin{tabular}{lllcllc}
    \midrule
    \multicolumn{3}{c}{$\Pers_*(\Fcal)$} & & \multicolumn{3}{c}{$\Pers_*(K/\Fcal)$} \\ 
    \cmidrule{1-3}\cmidrule{5-7}
    \makecell[c]{Type} & Interval & Dim & & Interval(s) & Dim & \makecell[c]{Type} \\ 
    \midrule 
    \raisebox{0pt}{\includegraphics[height=5.5pt]{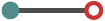}} closed-open &
    $[\birth,\death]$ & $\Dim$ & $\mapsto$ &
    $[\birth,\death]$ & $\Dim+1$ 
    & \raisebox{0pt}{\includegraphics[height=5.5pt]{fig/co}}
    \\\cmidrule{1-7} 
    \raisebox{0pt}{\includegraphics[height=5.5pt]{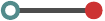}} open-closed &
    $[\birth,\death]$ & $\Dim$ & $\mapsto$ &
    $[\birth,\death]$ & $\Dim+1$ 
    & \raisebox{0pt}{\includegraphics[height=5.5pt]{fig/oc}}
    \\\cmidrule{1-7} 
    \raisebox{0pt}{\includegraphics[height=5.5pt]{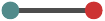}} closed-closed &
    $[\birth,\death]$ & $\Dim$ & $\mapsto$ &
    $[0,\birth-1]$, $[\death+1,\filtcnt]$ & $\Dim$ 
    & \raisebox{0pt}{\includegraphics[height=5.5pt]{fig/cc-rel}}
    \\\cmidrule{1-7} 
    \raisebox{0pt}{\includegraphics[height=5.5pt]{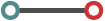}} open-open &
    $[\birth,\death]$ & $\Dim$ & $\mapsto$ &
    $[0,\death]$, $[\birth,\filtcnt]$ & $\Dim+1$ 
    & \raisebox{-2.5pt}{\includegraphics[width=\ccrelwidth]{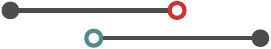}}
    \\
    \midrule
\end{tabular}}
\end{center}
By saying that the correspondence is surjective,
we mean that every interval in $\Pers_*(K/\Fcal)$
has a correspondence in $\Pers_*(\Fcal)$ according to the above rules.
\end{theorem}
\begin{remark}
Closed-closed intervals in $\Pers_*(\Fcal)$
correspond to intervals of the same dimension in $\Pers_*(K/\Fcal)$,
while for other types of intervals in $\Pers_*(\Fcal)$,
there is a dimension shift for the 
corresponding intervals in $\Pers_*(K/\Fcal)$.
Also,
each closed-closed or open-open interval in $\Pers_*(\Fcal)$
corresponds to \emph{two intervals} in $\Pers_*(K/\Fcal)$.
\end{remark}

A (strong) duality result (such as the one in~\cite{de2011dualities})
should be `bijective' in a sense that 
one is able to recover each side from the other.
However, this is not the case for the duality presented in Theorem~\ref{thm:dual}.
While one can derive $\Pers_*(K/\Fcal)$ from $\Pers_*(\Fcal)$,
there is an obstacle for recovering $\Pers_*(\Fcal)$ from $\Pers_*(K/\Fcal)$.
That is, in order to have the closed-closed and open-open intervals in $\Pers_*(\Fcal)$, 
one has to properly pair an interval in $\Pers_*(K/\Fcal)$
starting with 0 to an interval in $\Pers_*(K/\Fcal)$ ending with $\filtcnt$.
There is no obvious way to do so 
without knowing the {\it representatives} for these intervals~\cite{maria2014zigzag}.
That is the reason why we call this duality {\it weak}.
However, we show in Section~\ref{sec:Lefschetz} that
for certain intervals of $\Pers_*(K/\Fcal)$
when $K$ is a manifold,
the `reversed' pairing is indeed feasible without
the representatives.

\paragraph{Example.}

\begin{figure} 
    \begin{subfigure}{0.33\textwidth}
    \centering
    \includegraphics[width=0.85\textwidth]{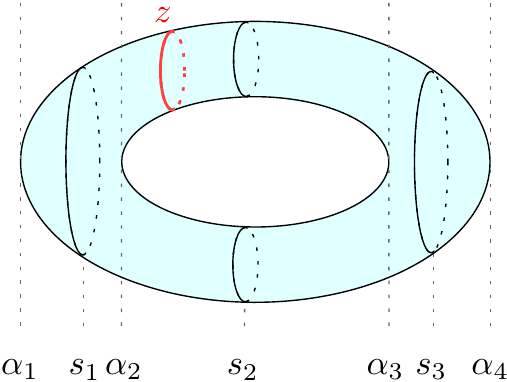} 
    \vspace{.5em}
    \caption{The torus $K$ with the height function $h$ taken over the horizontal line.}
    \label{fig:torus}
    \end{subfigure}
    \begin{subfigure}{0.33\textwidth}
    \centering
    \includegraphics[width=0.85\textwidth]{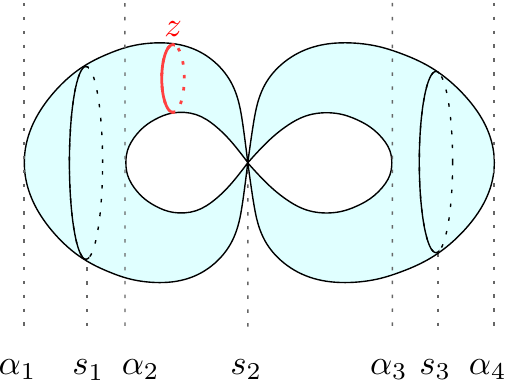} 
    \vspace{.5em}
    \caption{The space $|K|/|K_4|$.}
    \label{fig:contracted}
    \end{subfigure}
    \begin{subfigure}{0.33\textwidth}
    \centering
    \begin{tabular}{lcl}
        \makecell{$\Fcal$} & & \makecell{$K/\Fcal$} \\ 
        \midrule 
        $[1,7]^\text{cc}_0$ & $\mapsto$ & $[0,0]_0,[8,8]_0$ \\
        $[4,4]^\text{oo}_0$ & $\mapsto$ & $[0,4]_1,[4,8]_1$ \\
        $[2,6]^\text{oo}_1$ & $\mapsto$ & $[0,6]_2,[2,8]_2$ \\
        $[3,5]^\text{cc}_1$ & $\mapsto$ & $[0,2]_1,[6,8]_1$ \\
    \end{tabular}
    \vspace{1em}
    \caption{Interval mapping from $\Pers_*(\Fcal)$ to $\Pers_*(K/\Fcal)$.}
    \label{fig:ex-map}
    \end{subfigure}

    \begin{subfigure}{\textwidth}
    \centering
    \vspace{2em}
    \includegraphics[width=0.72\textwidth]{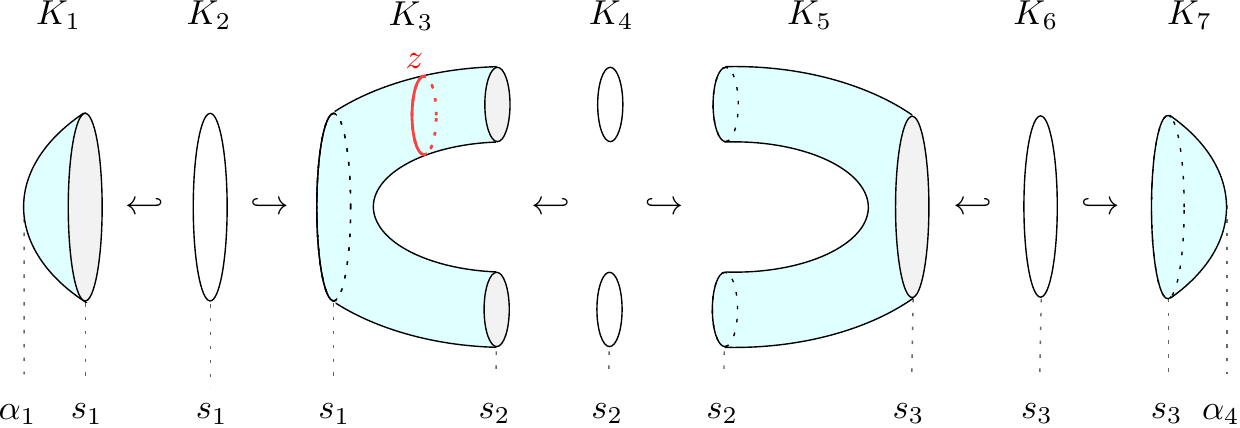} 
    \vspace{.5em}
    \caption{The levelset filtration $\Fcal$ of $h$ which also contains $K_0=\emptyset$ and $K_8=\emptyset$ at the two ends.}
    \label{fig:filt}
    \end{subfigure}

\caption{An example for the weak duality of absolute and relative zigzag.}
\label{fig:dual-ex} 
\end{figure}


In Figure~\ref{fig:dual-ex}, we provide an example for the duality described above,
where we take the height function $h$ on a torus $K$ (Figure~\ref{fig:torus})
and consider the levelset filtration $\Fcal$~\cite{carlsson2009zigzag-realvalue} 
of $h$ (Figure~\ref{fig:filt}).
We also assume that the torus is finely triangulated.
Let $-\infty=\aG_0<\aG_1<\cdots<\aG_4<\aG_5=\infty$ be the critical values of $h$
and $s_i$'s be the regular values s.t.\ $\aG_0<s_0<\aG_1<\cdots<\aG_4<s_4<\aG_5$.
Then,
the levelset filtration $\Fcal$ is 
defined with $h\inv(s_i)$ and $h\inv[s_i,s_{i+1}]$ for the regular values (see~\cite{carlsson2009zigzag-realvalue}).
The mapping of intervals from $\Pers_*(\Fcal)$ to $\Pers_*(K/\Fcal)$
as in Theorem~\ref{thm:dual} is listed in Figure~\ref{fig:ex-map}.
In Figure~\ref{fig:ex-map},
a superscript `cc' (resp.\ `oo') indicates that the interval
is closed-closed (resp.\ open-open)
and the subscripts denote the dimension of the intervals.

The reason for $[3,5]$ to form a closed-closed interval in $\Pers_1(\Fcal)$
is that the 1-cycle $z$ is created in $K_3$
and its homology class continues to exists in $K_4$ and $K_5$.
To see how a homology class gets born and dies with the
corresponding relative intervals,
observe that $[z]$ is non-trivial in $\Hm_1(K,K_0)=\Hm_1(K)$.
Also verify that $[z]$ is non-trivial in
$\Hm_1(K,K_1)$ and $\Hm_1(K,K_2)$.
However, $[z]$ becomes trivial in $\Hm_1(K,K_3)$, 
$\Hm_1(K,K_4)$, and $\Hm_1(K,K_5)$.
A quick way to see this is to utilize the fact that $\Hm_*(K,K_i)\simeq\Hmr_*(|K|/|K_i|)$
for the {\it good pair} $(K,K_i)$~\cite{hatcher2002algebraic}.
For example, we show $|K|/|K_4|$ in Figure~\ref{fig:contracted},
in which $z$ is a trivial 1-cycle.
Therefore, $[0,2]$ (and similarly $[6,8]$) forms an interval
in $\Pers_1(K/\Fcal)$ represented by $[z]$.

We can also look at 
the open-open
interval $[2,6]\in\Pers_1(\Fcal)$,
which is represented by the 1-cycle in $K_2$ and the 
homologous cycles in the remaining complexes.
Note that this interval $[2,6]$ is actually associated with the only non-trivial class
in $\Hm_2(K)$.
(Intuitively, one can imagine sweeping the 1-cycle in $K_2$ over
the horizontal line with proper forking and joining,
which would then obtain the entire surface $K$;
for more insight into this, we recommend the work~\cite{carlsson2009zigzag-realvalue} and~\cite{dey2021computing-levelsetcyc}.)
Describing
representatives for the corresponding intervals 
$[0,6]$ and $[2,8]$ in the relative module 
is not as straightforward. 
However,
the fact that 
these two form persistence intervals in $\Pers_2(K/\Fcal)$ 
is consistent with the following observation:
the 2nd Betti number $\bG_2$ of the pairs $(K,K_0)$, $(K,K_1)$,
$(K,K_7)$, and $(K,K_8)$ is 1, whereas $\bG_2$
of the remaining pairs from $(K,K_2)$ to $(K,K_6)$
is 2 (e.g., there are two 2-cycles in $|K|/|K_4|$ in Figure~\ref{fig:contracted}).
Explicitly listing representatives for the two intervals
is beyond the scope of this paper and one can refer to the work~\cite{maria2014zigzag}
and~\cite{dey2021computing-levelsetcyc}
for details.

\subsection{Justification}

To prove Theorem~\ref{thm:dual},
we draw upon the Diamond Principle proposed by Carlsson and de Silva~\cite{carlsson2010zigzag}
(see also~\cite{carlsson2009zigzag-realvalue}),
which relates the barcodes of two filtrations
differing by a local change. 
We first provide the following definition:
\begin{definition}[Mayer-Vietoris diamond~\cite{carlsson2010zigzag,carlsson2009zigzag-realvalue}]
\label{dfn:diamond}
Two simplex-wise filtrations $\Fcal$ and $\Fcal'$ 
are related by a \emph{Mayer-Vietoris diamond}
if they are of the following forms (where $\splx\neq\tG$):
\begin{equation}
\label{eqn:diamond}
\begin{tikzcd}[column sep=1.6em,
  row sep=0.4em,
]
\Fcal: &[-2em] & & & K_j\arrow[rd,hookleftarrow,pos=0.4,"\tG"]
\\
& K_0\arrow[r,leftrightarrow] & 
  \cdots\arrow[r,leftrightarrow] & 
  K_{j-1}\arrow[ur,hookrightarrow,pos=0.65,"\splx"]
  \arrow[dr,hookleftarrow,pos=0.4,"\tG"] 
  & & 
  K_{j+1}
  \arrow[r,leftrightarrow] & 
  \cdots\arrow[r,leftrightarrow] & 
  K_\filtcnt\\
\Fcal': & & & & K'_j\arrow[ru,hookrightarrow,pos=0.65,"\splx"]
\\
\end{tikzcd}
\end{equation}
In the above diagram,
$\Fcal$ and $\Fcal'$ differ only in the complexes at index $j$
and $\Fcal'$ is derived from $\Fcal$ by switching the addition of $\splx$ 
and deletion of $\tG$.
We also say that $\Fcal'$ is derived from $\Fcal$ by an \emph{outward} switch
and $\Fcal$ is derived from $\Fcal'$ by an \emph{inward} switch.
\end{definition}

\begin{remark}\label{rmk:diamond-invalid}
Note that the diagram in Equation~(\ref{eqn:diamond})
is invalid when $\splx=\tG$.
To see this, suppose that the two simplices equal.
Then,
the fact that $\tG=\splx$ is deleted from $K_{j-1}$ in $\Fcal'$ implies
that $\splx\in K_{j-1}$.
This makes $\Fcal$ invalid because we cannot add $\splx$ to $K_{j-1}$ anymore.
\end{remark}

\begin{remark}
In Equation~(\ref{eqn:diamond}),
we only provide a specific form of Mayer-Vietoris diamond
which is sufficient for our purposes;
see~\cite{carlsson2010zigzag,carlsson2009zigzag-realvalue} for a more general form.
According to~\cite{carlsson2010zigzag},
the diamond in Equation~(\ref{eqn:diamond})
is a Mayer-Vietoris diamond 
because $K_j=\linebreak[1]K_{j-1}\union K_{j+1}$
and $K'_j=K_{j-1}\intsec K_{j+1}$.
\end{remark}

We then have the following fact:

\begin{theorem}[Diamond Principle~\cite{carlsson2010zigzag}]
\label{thm:diamond}
Given two simplex-wise filtrations $\Fcal,\Fcal'$ 
related by a Mayer-Vietoris diamond
as in Equation~(\ref{eqn:diamond}),
there is a bijection from $\Pers_*(\Fcal)$ to $\Pers_*(\Fcal')$
as follows{\rm:}

\begin{center}
\begin{tabular}{lll}
     \midrule
     \makecell{$\Pers_*(\Fcal)$} & & \makecell{$\Pers_*(\Fcal')$} \\
     \midrule
     $[b,j-1]${\rm;} $b\leq j-1$ & $\mapsto$ & $[b,j]$ \\
     $[b,j]${\rm;} $b\leq j-1$ & $\mapsto$ & $[b,j-1]$ \\
     $[j,d]${\rm;} $d\geq j+1$ & $\mapsto$ & $[j+1,d]$ \\
     $[j+1,d]${\rm;} $d\geq j+1$ & $\mapsto$ & $[j,d]$ \\
     $[j,j]$ {\rm of dimension $\Dim$}  & $\mapsto$ & $[j,j]$ {\rm of dimension $\Dim-1$} \\
     $[b,d]${\rm; all other cases} & $\mapsto$ & $[b,d]$ \\
     \midrule
\end{tabular}
\end{center}
Note that
the bijection preserves the dimension of the intervals
except for $[j,j]$.
\end{theorem}
\begin{remark}\label{rmk:diamond}
In the above bijection, only an interval containing {\it some but not all}
of $\Set{j-1,j,j+1}$ maps to a different interval
or different dimension.
\end{remark}

We also observe the following fact which is critical for the proof
of Theorem~\ref{thm:dual}:

\begin{proposition}\label{prop:norep-filt-UD}
Let $K$ be a simplicial complex and 
$
\Fcal:\emptyset=K_0\leftrightarrow K_1\leftrightarrow
\cdots
\leftrightarrow K_\filtcnt=\emptyset
$
be a simplex-wise filtration of $K$
which is non-repetitive.
Then, there is a simplex-wise {\rm up-down} filtration
\[\ud:\emptyset=L_0\hookrightarrow L_1\hookrightarrow
\cdots\hookrightarrow 
L_{\simpcnt}=K\hookleftarrow
L_{\simpcnt+1}\hookleftarrow
\cdots
\hookleftarrow L_{2\simpcnt}=\emptyset\]
where $\filtcnt=2\simpcnt$
s.t.\ 
$\Fcal$ is derived from $\ud$ by a sequence of outward switches.
\end{proposition}

\begin{remark}
In Section~\ref{sec:comput-no-rept},
we further show
how an up-down filtration can be converted into a {\it non-zigzag} filtration
and hence relate the barcode computation for non-repetitive filtrations
to the computation of standard persistence.
\end{remark}

\begin{proof}
We prove an equivalent statement, i.e.,
$\ud$ is derived from $\Fcal$
by a sequence of inward switches.
Suppose that $\Fcal$ is of the form
$
\Fcal:
K_0\leftrightarrowsp{\fsimp{}{0}} K_1\leftrightarrowsp{\fsimp{}{1}}
\cdots 
\leftrightarrowsp{\fsimp{}{\filtcnt-1}} K_\filtcnt
$.
Let $K_{i}\bakinctosp{\fsimp{}{i}}K_{i+1}$ be the first deletion in $\Fcal$
and $K_{j}\inctosp{\fsimp{}{j}}K_{j+1}$ be the first addition after that.
That is, $\Fcal$ is of the form
\[
\Fcal:K_0\hookrightarrow \cdots\hookrightarrow
K_{i}\bakinctosp{\fsimp{}{i}} 
K_{i+1}\bakinctosp{\fsimp{}{i+1}}
\cdots
\bakinctosp{\fsimp{}{j-2}}
K_{j-1}\bakinctosp{\fsimp{}{j-1}}
K_{j}\inctosp{\fsimp{}{j}}K_{j+1}
\leftrightarrow\cdots \leftrightarrow K_\filtcnt
\]
Since $\Fcal$ is non-repetitive,
we have $\fsimp{}{j-1}\neq\fsimp{}{j}$.
So we can switch 
$\bakinctosp{\fsimp{}{j-1}}$
and $\inctosp{\fsimp{}{j}}$ 
(which is an inward switch) to derive a filtration
\[
K_0\hookrightarrow \cdots\hookrightarrow
K_{i}\bakinctosp{\fsimp{}{i}} 
K_{i+1}\bakinctosp{\fsimp{}{i+1}}
\cdots
\bakinctosp{\fsimp{}{j-2}}
K_{j-1}\inctosp{\fsimp{}{j}}
K'_{j}\bakinctosp{\fsimp{}{j-1}}
K_{j+1}
\leftrightarrow\cdots \leftrightarrow K_\filtcnt
\]
We then continue performing such inward switches (e.g., the next switch is on 
$\bakinctosp{\fsimp{}{j-2}}$
and $\inctosp{\fsimp{}{j}}$)
to derive a filtration
\[
\Fcal':K_0\hookrightarrow \cdots\hookrightarrow
K_{i}\inctosp{\fsimp{}{j}}
K'_{i+1}\bakinctosp{\fsimp{}{i}} 
\cdots
\bakinctosp{\fsimp{}{j-3}}
K'_{j-1}\bakinctosp{\fsimp{}{j-2}}
K'_{j}\bakinctosp{\fsimp{}{j-1}}
K_{j+1}
\leftrightarrow\cdots \leftrightarrow K_\filtcnt
\]
Note that from $\Fcal$ to $\Fcal'$, the up-down `prefix' grows longer.
We can repeat the above operations on the newly derived $\Fcal'$
until the entire filtration turns into an up-down one.
\end{proof}

We now prove Theorem~\ref{thm:dual}:
\begin{proof}[Proof of Theorem~\ref{thm:dual}]
Let 
$
\ud:\emptyset=K'_0\hookrightarrow 
\cdots\hookrightarrow K'_{\simpcnt}=K\hookleftarrow
\cdots
\hookleftarrow K'_{2\simpcnt}=\emptyset
$
be the simplex-wise, up-down filtration for $\Fcal$ 
as described in Proposition~\ref{prop:norep-filt-UD}.
We first prove that Theorem~\ref{thm:dual}
holds for $\ud$.
To prove this,
we decompose
$\ud$ into two non-zigzag filtrations 
$\ud_1,\ud_2$ sharing the complex $K'_\simpcnt=K$ as follows:
\[
\begin{tikzpicture}[xscale=2,yscale=1.0]
\draw (-0.7,0) node {$\ud:$};
\draw (0,0) node(00) {$K'_0$};
\draw (1,0) node(10) {$\cdots$};
\draw (2,0) node(20) {$K'_{\simpcnt-1}$};
\draw (3,0) node(30) {$K'_{\simpcnt}$};
\draw (4,0) node(40) {$K'_{\simpcnt+1}$};
\draw (5,0) node(50) {$\cdots$};
\draw (6,0) node(60) {$K'_{2\simpcnt}$};

\draw (-0.7,-1) node {$K/\ud:$};
\draw (0,-1) node(01) {$(K,K'_0)$};
\draw (1,-1) node(11) {$\cdots$};
\draw (2,-1) node(21) {$(K,K'_{\simpcnt-1})$};
\draw (3,-1) node(31) {$\emptyset$}; 
\draw (4,-1) node(41) {$(K,K'_{\simpcnt+1})$};
\draw (5,-1) node(51) {$\cdots$};
\draw (6,-1) node(61) {$(K,K'_{2\simpcnt})$};

\draw[right hook->] (00) edge (10);
\draw[right hook->] (10) edge (20);
\draw[right hook->] (20) edge (30);
\draw[left hook->] (40) edge (30);
\draw[left hook->] (50) edge (40);
\draw[left hook->] (60) edge (50);

\draw[right hook->] (01) edge (11);
\draw[right hook->] (11) edge (21);
\draw[right hook->] (21) edge (31);
\draw[left hook->] (41) edge (31);
\draw[left hook->] (51) edge (41);
\draw[left hook->] (61) edge (51);

\draw[densely dotted,|-latex]
(0,0.65) --
node [midway,above] {$\ud_1$}
(2.95,0.65);

\draw[densely dotted,|-latex]
(6,0.65) --
node [midway,above] {$\ud_2$}
(3.05,0.65);

\draw[densely dotted,|-latex]
(0,-1.65) --
node [midway,below] {$K/\ud_1$}
(2.95,-1.65);

\draw[densely dotted,|-latex]
(6,-1.65) --
node [midway,below] {$K/\ud_2$}
(3.05,-1.65);

\end{tikzpicture}
\]
Note that we denote $(K,K'_\simpcnt)$ as $\emptyset$
in $K/\ud$
because $\Hm_*(K,K'_\simpcnt)=0$.
We then verify Theorem~\ref{thm:dual} on $\ud$
by utilizing the duality of absolute and relative persistence in the non-zigzag setting~\cite{de2011dualities}. Because of the arrow directions, the birth and death indices
of a closed-open interval of $\Pers_\Dim(\ud)$ 
are also indices for $\ud_1$.
It follows that any closed-open interval $[b,d]\in\Pers_\Dim(\ud)$
is necessarily a {\it finite} interval in $\Pers_\Dim(\ud_1)$.
By the duality for non-zigzag persistence~\cite{de2011dualities},
$[b,d]\in\Pers_\Dim(\ud_1)$ corresponds to $[b,d]\in\Pers_{\Dim+1}(K/\ud_1)$.
Note that $[b,d]\in\Pers_{\Dim+1}(K/\ud_1)$
is also an interval in $\Pers_{\Dim+1}(K/\ud)$,
and hence the correspondence in Theorem~\ref{thm:dual} is satisfied.
Similarly, one can verify the correspondence for 
open-closed intervals in $\Pers_\Dim(\ud)$,
which are necessarily finite intervals in $\Pers_\Dim(\ud_2)$.
Now let $[b,d]\in\Pers_\Dim(\ud)$ be
a closed-closed interval.
Then, $b$ and $d$ must be in the
indices of $\ud_1$ and
$\ud_2$ respectively, 
thus inducing {\it infinite} intervals $[b,\simpcnt]\in\Pers_\Dim(\ud_1)$
and $[\simpcnt,d]\in\Pers_\Dim(\ud_2)$.
Note that while $\ud_2$ is a standard (non-zigzag) filtration
starting with $K'_{2\simpcnt}=\emptyset$ and growing into $K'_\simpcnt=K$,
we inherit the indexing from $\ud$ for $\ud_2$
which is thus an indexing in decreasing order.
Then,
by the duality for non-zigzag persistence~\cite{de2011dualities},
$[b,\simpcnt]\in\Pers_\Dim(\ud_1)$ corresponds to $[0,b-1]\in\Pers_\Dim(K/\ud_1)$
and $[\simpcnt,d]\in\Pers_\Dim(\ud_2)$ corresponds to $[d+1,2\simpcnt]\in\Pers_\Dim(K/\ud_2)$.
This implies the correspondence claimed in Theorem~\ref{thm:dual}.
Moreover, the correspondence as we have verified for $\Pers_*(\ud)$ and $\Pers_*(K/\ud)$
is surjective because the correspondence in~\cite{de2011dualities} is surjective.
Note that arrow directions in $\ud$
disallow any open-open intervals to exist in $\Pers_*(\ud)$.
Later in the proof,  we will see how open-open intervals are introduced
into $\Pers_*(\Fcal)$.

Next, we prove that the theorem holds for $\Fcal$ by induction.
By Proposition~\ref{prop:norep-filt-UD}, 
$\Fcal$ can be derived from $\ud$ by a sequence of outward switches.
Then, let $\Fcal^0=\ud,\Fcal^1,\ldots,\Fcal^{k-1},\Fcal^k=\Fcal$
be a sequence of filtrations such that each $\Fcal^{i+1}$ is derived from $\Fcal^{i}$
by an outward switch.
Inducting on $i$, we only need to prove that if Theorem~\ref{thm:dual} holds for $\Fcal^{i}$,
then it also holds for $\Fcal^{i+1}$. 

Without loss of generality, suppose that $\Fcal^{i}$ and $\Fcal^{i+1}$ are of the following forms:
\begin{equation*}
\begin{tikzcd}[column sep=1.6em,
  row sep=0.4em,
  /tikz/column 1/.append style={anchor=base west}]
\Fcal^i: &[-2em] & & & L_j\arrow[rd,hookleftarrow,pos=0.4,"\splx'"]
\\
& L_0\arrow[r,leftrightarrow] & 
  \cdots\arrow[r,leftrightarrow] & 
  L_{j-1}\arrow[ur,hookrightarrow,pos=0.65,"\splx"]
  \arrow[dr,hookleftarrow,pos=0.4,"\splx'"] 
  & & 
  L_{j+1}
  \arrow[r,leftrightarrow] & 
  \cdots\arrow[r,leftrightarrow] & 
  L_\filtcnt\\
\Fcal^{i+1}: & & & & L'_j\arrow[ru,hookrightarrow,pos=0.65,"\splx"]
\\
\end{tikzcd}
\end{equation*}
Note that $K/\Fcal^i$ and $K/\Fcal^{i+1}$ are also related by
a (general version of) Mayer-Vietoris diamond and the mapping of
$\Pers_*(K/\Fcal^i)$ and $\Pers_*(K/\Fcal^{i+1})$ is the same
as in Theorem~\ref{thm:diamond} (see~\cite{carlsson2009zigzag-realvalue}).
Assuming that Theorem~\ref{thm:dual} holds for $\Fcal^{i}$,
we need to verify that the theorem holds for $\Fcal^{i+1}$
under the following cases (recall Remark~\ref{rmk:inc-bd}):
\begin{enumerate}
    \item $L_{j-1}\hookrightarrow L_{j}$ introduces $j$ as 
    a {\it birth} index;
    $L_{j}\hookleftarrow L_{j+1}$ introduces $j+1$ as 
    a {\it birth} index.
    
    \item\label{itm:start-end-ind}
    $L_{j-1}\hookrightarrow L_{j}$ introduces $j$ as 
    a {\it birth} index;
    $L_{j}\hookleftarrow L_{j+1}$ introduces $j$ as 
    a {\it death} index.
    
    \item $L_{j-1}\hookrightarrow L_{j}$ introduces $j-1$ as 
    a {\it death} index;
    $L_{j}\hookleftarrow L_{j+1}$ introduces $j+1$ as 
    a {\it birth} index.
    
    \item $L_{j-1}\hookrightarrow L_{j}$ introduces $j-1$ as 
    a {\it death} index;
    $L_{j}\hookleftarrow L_{j+1}$ introduces $j$ as 
    a {\it death} index.
\end{enumerate}

We only verify Case~\ref{itm:start-end-ind} 
and omit the verification for other cases
which is similar.
Now suppose that Case~\ref{itm:start-end-ind} happens.
We define the following maps:
\begin{equation*}
\begin{tikzcd}[column sep=1em,row sep=1.5em]
\Pers_*(\Fcal^{i})\arrow[r,"\phi"]\arrow[d,"\psi",swap,pos=0.4] & 
\Pers_*(K/\Fcal^{i})\arrow[d,"\psi",pos=0.4]
\\
\Pers_*(\Fcal^{i+1})
& 
\Pers_*(K/\Fcal^{i+1})
\end{tikzcd}
\end{equation*}
in which:
\begin{itemize}
    \item $\phi$ denotes the mapping of intervals 
    from $\Pers_*(\Fcal^i)$ to $\Pers_*(K/\Fcal^i)$
    by Theorem~\ref{thm:dual}.
    Note that
    for an interval $[b,d]$ in $\Pers_*(\Fcal^i)$,
    $\phi([b,d])$ is a {\it set} of one or two intervals.
    
    \item $\psi$ denotes the bijection from $\Pers_*(\Fcal^i)$ to $\Pers_*(\Fcal^{i+1})$
    and the bijection from $\Pers_*(K/\Fcal^i)$ to $\Pers_*(K/\Fcal^{i+1})$
    by the Diamond Principle (Theorem~\ref{thm:diamond}).
\end{itemize}

To prove that Theorem~\ref{thm:dual} holds for $\Fcal^{i+1}$
under Case~\ref{itm:start-end-ind},
we show that
for any interval $[b,d]$ in $\Pers_*(\Fcal^i)$, the following mappings
\begin{equation*}
\begin{tikzcd}[column sep=1em,row sep=1.5em]
{[}b,d{]}\in\Pers_*(\Fcal^{i})\arrow[r,mapsto,"\phi"]\arrow[d,mapsto,swap,"\psi",pos=0.4] & 
\phi({[}b,d{]})\subseteq\Pers_*(K/\Fcal^{i})\arrow[d,mapsto,"\psi",pos=0.4]
\\
\psi({[}b,d{]})\in\Pers_*(\Fcal^{i+1})
& 
\psi(\phi({[}b,d{]}))\subseteq\Pers_*(K/\Fcal^{i+1})
\end{tikzcd}
\end{equation*}
satisfy that the interval $\psi([b,d])$ corresponds to the intervals $\psi(\phi([b,d]))$ as claimed
in Theorem~\ref{thm:dual}.
Therefore, the correspondence in Theorem~\ref{thm:dual} is verified for $\Fcal^{i+1}$.
Furthermore,
the above fact implies the following:
given that the correspondence for $\Fcal^{i}$ is surjective,
the correspondence for $\Fcal^{i+1}$ is also surjective.

For an interval $[b,d]\in\Pers_*(\Fcal^i)$
containing {\it all or none} of $\Set{j-1,j,j+1}$,
we note that 
intervals in $\phi([b,d])$
also contain all or none of $\Set{j-1,j,j+1}$.
For example, if $[b,d]$ is a closed-closed interval in $\Pers_*(\Fcal^i)$
s.t.\ $d<j-1$, then $\phi([b,d])=\Set{[0,b-1],[d+1,\filtcnt]}$,
where $[0,b-1]$ contains none of $\Set{j-1,j,j+1}$
and $[d+1,\filtcnt]$ contains all of them.
Hence, by Remark~\ref{rmk:diamond}, 
$\psi([b,d])=[b,d]$,
$\psi(\phi([b,d]))=\phi([b,d])$,
and the intervals' dimensions are preserved. Furthermore,
the types (i.e., `open' or `closed') of the two ends of $[b,d]$
do not change from $\Pers_*(\Fcal^i)$ to $\Pers_*(\Fcal^{i+1})$.
So indeed $\psi(\phi([b,d]))$ are the intervals
that $\psi([b,d])$ corresponds to as stated in Theorem~\ref{thm:dual}.

We then show that for any interval in $\Pers_*(\Fcal^i)$
containing {\it some but not all} of $\Set{j-1,j,j+1}$,
the correspondence is still correct
after the switch. We have the following two cases:

\begin{itemize}
    \item 
First suppose that 
$[j,j]$ does not form an interval in $\Pers_*(\Fcal^i)$.
Then, according to the assumptions in Case~\ref{itm:start-end-ind}, the intervals in $\Pers_*(\Fcal^i)$ 
containing some but not all of $\Set{j-1,j,j+1}$
are $[b,j]$ and $[j,d]$, where $b\leq j-1$ and $d\geq j+1$.
If $[b,j]$ is a $\Dim$-th closed-closed interval in $\Pers_*(\Fcal^i)$,
we have the following mappings:
\begin{equation*}
\begin{tikzcd}[column sep=1em,row sep=1.5em]
{[}b,j{]}\in\Pers_\Dim(\Fcal^{i})\arrow[r,mapsto,"\phi"]\arrow[d,mapsto,swap,"\psi",pos=0.4] & 
{[}0,b-1{]},{[}j+1,m{]}\in\Pers_\Dim(K/\Fcal^{i})\arrow[d,mapsto,"\psi",pos=0.4]
\\
{[}b,j-1{]}\in\Pers_\Dim(\Fcal^{i+1})
& 
{[}0,b-1{]},{[}j,m{]}\in\Pers_\Dim(K/\Fcal^{i+1})
\end{tikzcd}
\end{equation*}
where $[0,b-1]$ and $[j,m]$ of dimension $\Dim$ are exactly intervals in $\Pers_*(K/\Fcal^{i+1})$
that $[b,j-1]\in\Pers_\Dim(\Fcal^{i+1})$ (which is also closed-closed) corresponds to.

We can similarly verify the correctness of correspondence
when $[b,j]\in\Pers_*(\Fcal^i)$ is open-closed
(note that $j$ must be a closed death index)
and for $[j,d]\in\Pers_*(\Fcal^i)$. 

\item
Now suppose that 
$[j,j]$ forms a $\Dim$-th interval in $\Pers_*(\Fcal^i)$.
Then, again because of the assumptions in Case~\ref{itm:start-end-ind}, $[j,j]$ is the only interval in $\Pers_*(\Fcal^i)$ 
containing some but not all of $\Set{j-1,j,j+1}$.
Since $[j,j]\in\Pers_\Dim(\Fcal^i)$ is closed-closed, we have the following mappings:
\begin{equation*}
\begin{tikzcd}[column sep=1em,row sep=1.5em]
{[}j,j{]}\in\Pers_\Dim(\Fcal^{i})\arrow[r,mapsto,"\phi"]\arrow[d,mapsto,swap,"\psi",pos=0.4] & 
{[}0,j-1{]},{[}j+1,m{]}\in\Pers_\Dim(K/\Fcal^{i})\arrow[d,mapsto,"\psi",pos=0.4]
\\
{[}j,j{]}\in\Pers_{\Dim-1}(\Fcal^{i+1})
& 
{[}0,j{]},{[}j,m{]}\in\Pers_\Dim(K/\Fcal^{i+1})
\end{tikzcd}
\end{equation*}
where $[j,j]$ becomes open-open in $\Pers_{\Dim-1}(\Fcal^{i+1})$
and its correspondence in $\Pers_*(K/\Fcal^{i+1})$
illustrated above
is exactly as specified in Theorem~\ref{thm:dual}.
Notice that the above transition is the only time when
open-open intervals are introduced into the absolute modules, 
given that $\Pers_*(\ud)$ contains no open-open intervals
initially.
\qedhere
\end{itemize}
\end{proof}

\begin{remark}
The reason why our proof of Theorem~\ref{thm:dual} 
does not work for repetitive filtrations 
is that Proposition~\ref{prop:norep-filt-UD}
does not hold anymore. 
Suppose that there is a simplex $\sG$ in $\Fcal$ which is added again after being deleted.
Then, the sequence of inward switches
used to turn $\Fcal$ into $\ud$
in the proof of Proposition~\ref{prop:norep-filt-UD} 
must contain a switch of the deletion of $\sG$
with the addition of $\sG$, which is invalid due to Remark~\ref{rmk:diamond-invalid}.
\end{remark}




\section{Codimension-zero and -one zigzag for manifolds}
\label{sec:Lefschetz}

In this section, we first present a near-linear algorithm for computing
the $\Dim$-th relative barcode of a zigzag filtration
on a $\Dim$-manifold.
The algorithm utilizes Lefschetz duality~\cite{cohen2009extending,munkres2018elements}
to convert the problem into computing the 0-th barcode
of a {\it dual} filtration.
We then utilize the near-linear algorithm for 0-dimensional zigzag~\cite{dey2021graph}
to achieve the time complexity.
Assuming further that the filtration is non-repetitive,
the weak duality theorem in Section~\ref{sec:duality} then
gives an algorithm for computing the $\Dim$-th barcode
and (a subset of) the $(\Dim-1)$-th barcode for the absolute homology
in near-linear time.

Throughout the section,
we assume the following input:
\begin{itemize} \item 
$K$ is a simplicial 
$\Dim$-manifold
with $\simpcnt$ 
simplices
    and $\Fcal:\emptyset=K_0\leftrightarrowsp{\sG_0} K_1\leftrightarrowsp{\sG_1}
\cdots
\leftrightarrowsp{\sG_{\filtcnt-1}} K_\filtcnt=\emptyset
$
is a 
simplex-wise filtration 
of $K$.
\end{itemize}

Let $G$ denote the {\it dual graph} of $K$,
where vertices of $G$ bijectively map to $\Dim$-simplices of $K$
and edges of $G$ bijectively map to $(\Dim-1)$-simplices of $K$.
Define a {\it dual filtration} 
$\dfilt: G_0 \leftrightarrow G_1 \leftrightarrow \cdots \leftrightarrow G_\filtcnt$,
where each $G_i$ is a subgraph of $G$
s.t.\ a vertex (resp.\ edge) of $G$ is in $G_i$ iff
its dual $\Dim$-simplex (resp.\ $(\Dim-1)$-simplex) is {\it not} in $K_i$.
Each $G_i$ is a well-defined subgraph
because if a $(\Dim-1)$-simplex of $K$ is not in $K_i$,
all its $\Dim$-cofaces are also not in $K_i$.
Accordingly, the vertices of each edge of $G_i$ are in $G_i$.
Intuitively, $G_i$ encodes the connectivity of $|K|-|K_i|$
and an example of $\dfilt$ is given in~\cite[Section 5]{dey2021graph}
where $\Real^2$ is viewed as its one-point compactification $\mathbb{S}^2$.
Note that similarly as in~\cite[Section 5]{dey2021graph}, 
inclusion directions in $\dfilt$ are reversed
and $\dfilt$ is not necessarily simplex-wise
because an arrow may introduce no changes.

We now provide the first conclusion of this section:
\begin{theorem}\label{thm:kbyf-dfilt-eq}
$\Pers(\Hm_\Dim(K/\Fcal))=\Pers(\Hm_0(\dfilt))$.
Hence, $\Pers(\Hm_\Dim(K/\Fcal))$ can be computed with time complexity
$O(\filtcnt\log^2 \simpcnt\allowbreak+\filtcnt\log \filtcnt)$.
\end{theorem}

Assuming that $\Pers(\Hm_\Dim(K/\Fcal))=\Pers(\Hm_0(\dfilt))$,
we first construct the dual graph $G$ and the dual filtration $\dfilt$ in linear time.
After this, we compute $\Pers(\Hm_0(\dfilt))$ using the algorithm in~\cite{dey2021graph}.
The time complexity then follows from the time complexity for computing
the 0-th barcode~\cite{dey2021graph}.

The proof of $\Pers(\Hm_\Dim(K/\Fcal))=\Pers(\Hm_0(\dfilt))$
is by combining
the conclusions of Proposition~\ref{prop:M1-M2-eq},
Proposition~\ref{prop:kbyf-M1-eq},
and Proposition~\ref{prop:M2-dfilt-eq}.
Before presenting the propositions,
we first provide the (natural version of) Lefschetz duality~\cite[$\mathsection$72]{munkres2018elements}:
\begin{theorem}[Lefschetz duality]
\label{thm:lefschetz}
Let $k$ be an integer s.t.\ $0\leq k\leq\Dim$.
For the complex $K$ {\rm(}which is a $\Dim$-manifold{\rm)},
one can assign each subcomplex $L$ of $K$ an isomorphism 
$\lG_L:\Hm^k(K,L)\to\Hm_{\Dim-k}(|K|-|L|)$.
Moreover, the assignment is natural w.r.t.\ inclusion, 
i.e., for another subcomplex $L'\subseteq L$,
the following diagram
commutes,
where the horizontal maps are induced by inclusion:
\begin{equation*}
\begin{tikzcd}[column sep=1em,row sep=1.5em]
\Hm^k(K,L)\arrow[r]\arrow{d}{\simeq}[swap]{\lG_{L}} 
  & \Hm^k(K,L')\arrow{d}{\lG_{L'}}[swap]{\simeq}\\
\Hm_{\Dim-k}(|K|-|L|)\arrow[r] 
  & \Hm_{\Dim-k}(|K|-|L'|)
\end{tikzcd}
\end{equation*}

\end{theorem}

\begin{proposition}\label{prop:M1-M2-eq}
The following zigzag modules are isomorphic:
\begin{equation*}
\begin{tikzcd}[column sep=1em,row sep=1.5em,     
/tikz/column 1/.append style={anchor=base west}]
\Mcal_1: & \Hom(\Hm_\Dim(K,K_0))\arrow[r,leftrightarrow]\arrow[d,pos=0.4,"\simeq"] 
  & \Hom(\Hm_\Dim(K,K_1)) \arrow[r,leftrightarrow]\arrow[d,pos=0.4,"\simeq"] 
  & \cdots \arrow[r,leftrightarrow] 
  & \Hom(\Hm_\Dim(K,K_\filtcnt))\arrow[d,pos=0.4,"\simeq"]\\
\Hm^\Dim(K/\Fcal): & \Hm^\Dim(K,K_0)\arrow[r,leftrightarrow]\arrow[d,pos=0.4,"\simeq"] 
  & \Hm^\Dim(K,K_1) \arrow[r,leftrightarrow]\arrow[d,pos=0.4,"\simeq"] 
  & \cdot \arrow[r,leftrightarrow] 
  & \Hm^\Dim(K,K_\filtcnt) \arrow[d,pos=0.4,"\simeq"] \\
\Mcal_2: & \Hm_0(|K|-|K_0|)\arrow[r,leftrightarrow] 
  & \Hm_0(|K|-|K_1|)\arrow[r,leftrightarrow] 
  & \cdots \arrow[r,leftrightarrow] 
  & \Hm_0(|K|-|K_\filtcnt|)
\end{tikzcd}
\end{equation*}
which implies that $\Pers(\Mcal_1)=\Pers(\Mcal_2)$.
\end{proposition}
\begin{proof}
The isomorphism between $\Mcal_1$ and $\Hm^\Dim(K/\Fcal)$
is given by the universal coefficient theorem~\cite[pg.\ 196,198]{hatcher2002algebraic},
and 
the isomorphism between $\Hm^\Dim(K/\Fcal)$ and $\Mcal_2$
is given by the Lefschetz duality (Theorem~\ref{thm:lefschetz}).
\end{proof}

\begin{proposition}\label{prop:kbyf-M1-eq}
$\Pers(\Hm_\Dim(K/\Fcal))=\Pers(\Mcal_1)$.
\end{proposition}
\begin{proof}
The proof is the same as the proof of Proposition 24
in the full version of~\cite{dey2021graph}.
\end{proof}

\begin{proposition}\label{prop:M2-dfilt-eq}
$\Pers(\Mcal_2)=\Pers(\Hm_0(\dfilt))$.
\end{proposition}
\begin{proof}
The proof is the same as the proof of Proposition 25
in the full version of~\cite{dey2021graph}.
\end{proof}

We now assume that $\Fcal$ is non-repetitive,
by which
we can draw upon the weak duality theorem in Section~\ref{sec:duality}.
Our goal is to (partially) recover intervals for $\Pers(\Hm_{*}(\Fcal))$
from $\Pers(\Hm_\Dim(K/\Fcal))$ based on the weak duality theorem,
so that these intervals for $\Pers(\Hm_{*}(\Fcal))$
can be computed in near-linear time.
Our conclusion is as follows:

\begin{theorem}\label{thm:recover-F-from-KbyF}
Suppose that $\Fcal$ is non-repetitive.
Given $\Pers(\Hm_\Dim(K/\Fcal))$,
one can compute $\Pers(\Hm_\Dim(\Fcal))$,
and the closed-open, open-closed, and open-open intervals of 
$\Pers(\Hm_{\Dim-1}(\Fcal))$
in linear time.
Thus, the above subset of $\Pers(\Hm_*(\Fcal))$ can
be computed in
$O(\filtcnt\log^2 \simpcnt\allowbreak+\filtcnt\log \filtcnt)=
O(\simpcnt\log^2 \simpcnt)$ time, where $\filtcnt=2\simpcnt$.
\end{theorem}
\begin{remark}
The subset of $\Pers(\Hm_*(\Fcal))$ that we can recover
contain all those intervals derived from pairings involving $\Dim$-simplices. 
For example, an open-open interval of
$\Pers(\Hm_{\Dim-1}(\Fcal))$ is produced by pairing the deletion 
of a $\Dim$-simplex with the addition of another $\Dim$-simplex.
\end{remark}
\begin{proof}[Proof $($Sketch$)$]
(See Appendix~\ref{sec:pf-thm-recover-F-from-KbyF}
for the complete proof.)
By the weak duality theorem, we can directly recover all 
closed-open and open-closed intervals in $\Pers(\Hm_{\Dim-1}(\Fcal))$
from $\Pers(\Hm_\Dim(K/\Fcal))$.
To pair an interval $[0,d]\in\Pers_\Dim(K/\Fcal)$
with an interval $[b,\filtcnt]\in\Pers_\Dim(K/\Fcal)$ 
(see the discussions after Theorem~\ref{thm:dual} in Section~\ref{sec:duality}),
we note that the {\it destroyer} $\sG_d$ of $[0,d]$ 
and the {\it creator} $\sG_{b-1}$ of $[b,\filtcnt]$
(see Definition~\ref{dfn:creator-destroyer})
must come from the same connected component of $K$.
Also, there is exactly one
interval in $\Pers_\Dim(K/\Fcal)$
with birth at $0$ (death at $m$ resp.) per component of $K$ because
$\Hm_\Dim(K,K_0)=\Hm_\Dim(K)=\Hm_\Dim(K,K_m)$ has dimension equal to the
number of components in $K$.
We can then recover 
closed-closed intervals in $\Pers(\Hm_{\Dim}(\Fcal))$ 
and open-open intervals in $\Pers(\Hm_{\Dim-1}(\Fcal))$.
Note that 
$\Pers(\Hm_{\Dim}(\Fcal))$ contains only closed-closed intervals.
\end{proof}

\section{Non-repetitive zigzag via standard persistence}
\label{sec:comput-no-rept}

Proposition~\ref{prop:norep-filt-UD} presented in Section~\ref{sec:duality}
indicates that a {\it non-repetitive} filtration
admits an {\it up-down} filtration as its `canonical' form.
In this section,
utilizing this canonical form, we show that computing 
absolute barcodes (and hence relative barcodes
by duality) for non-repetitive zigzag filtrations 
can be reduced to computing barcodes
for certain {\it non-zigzag} filtrations.
This finding leads to a more efficient persistence algorithm 
for non-repetitive zigzag filtrations considering that standard (non-zigzag) persistence
admits faster algorithms~\cite{BDM15,CK13,de2011dualities} 
in practice, which we confirm with our experiments.
Furthermore, the finding locates instances of zigzag filtrations
which make the barcode computation far more slower
than the non-zigzag ones. These filtrations are those
where simplices are repeatedly added and deleted.

Throughout the section,
suppose that we are given a non-repetitive, simplex-wise filtration
\[\Fcal:
\emptyset=
K_0\leftrightarrowsp{\fsimp{}{0}} K_1\leftrightarrowsp{\fsimp{}{1}}
\cdots 
\leftrightarrowsp{\fsimp{}{\filtcnt-1}} K_\filtcnt
=\emptyset
\]
of a complex $K$.
Then,
let
\[\ud:\emptyset=\ucplx_0\inctosp{\usimp_0} 
\cdots\inctosp{\usimp_{\simpcnt-1}} 
\ucplx_{\simpcnt}=K\bakinctosp{\usimp_{\simpcnt}}
\cdots
\bakinctosp{\usimp_{2\simpcnt-1}} \ucplx_{2\simpcnt}=\emptyset\]
be the up-down filtration for $\Fcal$ 
as described in Proposition~\ref{prop:norep-filt-UD},
where $\filtcnt=2\simpcnt$.
Utilizing Proposition~\ref{prop:norep-filt-UD} and
the Diamond Principle (Theorem~\ref{thm:diamond}),
we first relate intervals of $\Pers_*(\Fcal)$ to intervals of $\Pers_*(\ud)$.
Then, we relate $\Pers_*(\ud)$ to the barcode of an 
{\it extended persistence}~\cite{cohen2009extending} filtration
and draw upon the efficient algorithms for
non-zigzag persistence~\cite{BDM15,CK13,de2011dualities}. 

In summary, our main
tasks are:
\begin{itemize}
    \item Convert the given non-repetitive filtration to an up-down filtration.
    \item Convert the up-down filtration to a non-zigzag filtration 
    with the help of an extended persistence filtration
    and compute the standard persistence barcode.
    \item Convert the standard persistence barcode to 
    the barcode of the input filtration based on rules
    given in Proposition~\ref{prop:ud-f-map} and~\ref{prop:ef-ud-map}.
\end{itemize}

\subsection{Conversion to up-down filtration}
In a simplex-wise zigzag filtration,
for a simplex $\sG$, let its addition be denoted as $\add{\sG}$
and its deletion be denoted as $\del{\sG}$.
From the proof of Proposition~\ref{prop:norep-filt-UD},
we observe the following: during the transition from $\Fcal$ to $\ud$,
for any two additions $\add{\sG}$, $\add{\sG'}$ in $\Fcal$ 
(and similarly for deletions),
if $\add{\sG}$ is before $\add{\sG'}$ in $\Fcal$,
then $\add{\sG}$ is also before $\add{\sG'}$ in $\ud$.
We then have the following fact:
\begin{fact}
Given the filtration $\Fcal$, to derive $\ud$, 
one only needs to scan $\Fcal$
and list all the additions first
and then the deletions, following the order in $\Fcal$.
\end{fact}
\begin{remark}
Figure~\ref{fig:F-ud} gives an example of $\Fcal$ and its corresponding $\ud$,
where the additions and deletions in $\Fcal$ and $\ud$
follow the same order.
\end{remark}

\begin{definition}[Creator and destroyer]
\label{dfn:creator-destroyer}
For any interval $[b,d]\in\Pers_*(\Fcal)$, if $K_{b-1}\leftrightarrowsp{\fsimp_{b-1}}K_b$
is forward  (resp.\ backward), 
we call $\add{\fsimp_{b-1}}$ (resp.\ $\del{\fsimp_{b-1}}$) the {\it creator}
of $[b,d]$. 
Similarly, if $K_{d}\leftrightarrowsp{\fsimp_{d}}K_{d+1}$
is forward (resp.\ backward), 
we call $\add{\fsimp_{d}}$ (resp.\ $\del{\fsimp_{d}}$) the {\it destroyer}
of $[b,d]$.
\end{definition}

By inspecting the interval mapping in the Diamond Principle,
we have the following fact:
\begin{proposition}\label{prop:creat-destr-same}
For two simplex-wise filtrations $\Lcal,\Lcal'$ 
related by a Mayer-Vietoris diamond,
any two intervals of $\Pers_*(\Lcal)$ and $\Pers_*(\Lcal')$
mapped by the Diamond Principle
have the same set of creator and destroyer,
though the creator and destroyer may swap.
This implies that there is a bijection from $\Pers_*(\ud)$ to $\Pers_*(\Fcal)$
s.t.\ any two corresponding intervals
have the same set of creator and destroyer.
\end{proposition}
\begin{remark}
The only time when the creator and destroyer swap
in a Mayer-Vietoris diamond
is when the interval $[j,j]$ for the upper filtration in Equation~(\ref{eqn:diamond}) 
turns into the same interval (of one dimension lower) for the lower filtration.
\label{rem:critical}
\end{remark}

Consider the example in Figure~\ref{fig:F-ud}
for an illustration of Proposition~\ref{prop:creat-destr-same}.
In the example,
$[1,2]\in\Pers_1(\Fcal)$ corresponds to $[1,4]\in\Pers_1(\ud)$, 
where their creator is $\add{\Set{a,d}}$ and their destroyer is $\del{\Set{a,b}}$.
Moreover, $[4,6]\in\Pers_0(\Fcal)$ corresponds to $[4,5]\in\Pers_1(\ud)$.
The creator of $[4,6]\in\Pers_0(\Fcal)$ is $\del{\Set{b,d}}$
and the destroyer is $\add{\Set{b,c}}$.
Meanwhile, $[4,5]\in\Pers_1(\ud)$ has the same set of creator and destroyer
but the roles swap.

\begin{figure}
  \centering
  \includegraphics[width=0.9\linewidth]{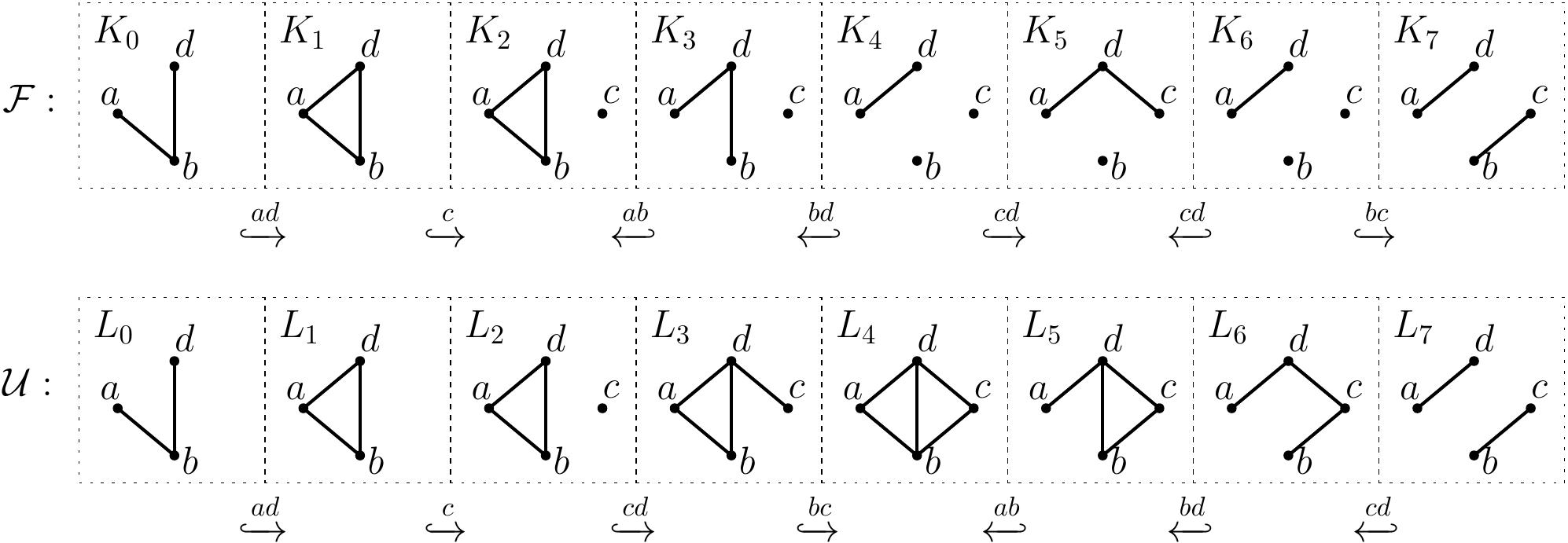}
  \caption{A non-repetitive filtration $\Fcal$ and its corresponding up-down filtration $\ud$.
  For brevity, $\Fcal$ does not start and end with empty complexes
  (which can be treated as a truncated case).}
  \label{fig:F-ud}
\end{figure}

For any $\add{\sG}$ or $\del{\sG}$ in $\Fcal$,
let $\id_\Fcal(\add{\sG})$ or $\id_\Fcal(\del{\sG})$ 
denote the index (position) of the addition or deletion.
For example, for an addition $K_i\inctosp{\fsimp{}{i}}K_{i+1}$ in $\Fcal$,
$\id_\Fcal(\add{\fsimp{}{i}})=i$.
Proposition~\ref{prop:creat-destr-same} indicates
the following explicit mapping from $\Pers_*(\ud)$ to $\Pers_*(\Fcal)$:
\begin{proposition}\label{prop:ud-f-map}
There is a bijection from $\Pers_*(\ud)$ to $\Pers_*(\Fcal)$
which maps each $[b,d]\in\Pers_\Dim(\ud)$
by the following rule{\rm:}
{\rm \begin{center}
\begin{tabular}{cccccl}
\midrule
\makecell{Type} & Condition & & Type & Interval in $\Pers_*(\Fcal)$ & Dim  \\
\midrule
closed-open & - & $\mapsto$ & closed-open &
  $\big[\id_\Fcal(\add{\usimp_{b-1}})+1,\id_\Fcal({\add{\usimp_d}})\big]$ &
  $\Dim$ \\ 
\cmidrule{1-6}
open-closed & - & $\mapsto$ & open-closed & 
  $\big[\id_\Fcal(\del{\usimp_{b-1}})+1,\id_\Fcal({\del{\usimp_d}})\big]$ &
  $\Dim$ \\ 
\cmidrule{1-6}
\multirow{2}{*}{closed-closed} &
  $\id_\Fcal(\add{\usimp_{b-1}})<\id_\Fcal({\del{\usimp_d}})$ &
  $\mapsto$ & closed-closed &
  $\big[\id_\Fcal(\add{\usimp_{b-1}})+1,\id_\Fcal({\del{\usimp_d}})\big]$ &
  $\Dim$ \\
& $\id_\Fcal(\add{\usimp_{b-1}})>\id_\Fcal({\del{\usimp_d}})$ &
  $\mapsto$ & open-open &
  $\big[\id_\Fcal({\del{\usimp_d}})+1,\id_\Fcal(\add{\usimp_{b-1}})\big]$ &
  $\Dim-1$ \\
\midrule
\end{tabular}
\end{center}}
\end{proposition}
\begin{remark}
As mentioned in the proof of Theorem~\ref{thm:dual},
$\Pers_*(\ud)$ contains no open-open intervals. 
However, a closed-closed interval $[b,d]\in\Pers_\Dim(\ud)$ turns into an open-open
interval in $\Pers_{\Dim-1}(\Fcal)$ when $\id_\Fcal(\add{\usimp_{b-1}})>\id_\Fcal({\del{\usimp_d}})$.
The type of such a closed-closed interval
changes 
when it
turns into a single point interval $[j,j]$ 
during the sequence of outward switches,
after which
the closed-closed interval $[j,j]$ becomes an open-open interval
$[j,j]$ with a dimension shift (see Theorem~\ref{thm:diamond}).
\end{remark}

We can take the example in Figure~\ref{fig:F-ud} for the mapping in
Proposition~\ref{prop:ud-f-map}. 
The interval $[4,5]\in\Pers_1(\ud)$ is a closed-closed one
whose creator is $\add{\Set{b,c}}$ and destroyer is $\del{\Set{b,d}}$.
We have that $\id_\Fcal(\add{\Set{b,c}})=6>\id_\Fcal(\del{\Set{b,d}})=3$,
so that the corresponding interval in $\Pers_0(\Fcal)$ is
\[
[\id_\Fcal(\del{\Set{b,d}})+1,\id_\Fcal(\add{\Set{b,c}})]=[4,6]
\]

\subsection{Conversion to non-zigzag filtration}
We first convert the up-down filtration $\ud$ to
an extended persistence~\cite{cohen2009extending} filtration $\ef$ which
is then easily converted to an (absolute) non-zigzag filtration using the `coning' technique~\cite{cohen2009extending}.

Inspired by the Mayer-Vietoris pyramid in~\cite{carlsson2009zigzag-realvalue},
we relate $\Pers_*(\ud)$ to the barcode of the filtration $\ef$
 defined as:
\[\ef:\emptyset=\ucplx_0
\incto
\cdots
\incto
\ucplx_{\simpcnt}=(K,\ucplx_{2\simpcnt})
\incto
(K,\ucplx_{2\simpcnt-1})
\incto
\cdots
\incto
(K,\ucplx_{\simpcnt})=(K,K)\]
where $\ucplx_{\simpcnt}=K=(K,\ucplx_{2\simpcnt}=\emptyset)$.
When denoting the persistence intervals of $\ef$,
we let the increasing index for the first half of $\ef$ continue to the second half,
i.e., $(K,L_{2\simpcnt-1})$ has index $\simpcnt+1$
and $(K,L_{\simpcnt})$ has index $2\simpcnt$.
Then, it can be verified that 
an interval $[b,d]\in\Pers_*(\ef)$ for $b<\simpcnt<d$
starts with the complex $L_b$ and ends with $(K,L_{3\simpcnt-d})$.

\begin{remark}
A filtration in extended persistence~\cite{cohen2009extending} 
is originally defined for a PL function $f$, where the first half
is the lower-star filtration of $f$
and the second half (the relative part)
is derived from the upper-star filtration of $f$.
The filtration $\ef$ defined above is a generalization of the one in~\cite{cohen2009extending}.
\end{remark}

\begin{proposition}\label{prop:ef-ud-map}
There is a bijection from $\Pers_*(\ef)$ to $\Pers_*(\ud)$
which maps each $[b,d]\in\Pers_*(\ef)$ of dimension $\Dim$
by the following rule{\rm:}
{\rm \begin{center}
\begin{tabular}{llclll}
\midrule
Type & Condition & & Type & Interv. in $\Pers_*(\ud)$ & Dim  \\
\midrule
Ord &
$d<\simpcnt$ & $\mapsto$ & 
  closed-open &
  $[b,d]$ &
  $\Dim$ \\ 
Rel &
$b>\simpcnt$ & $\mapsto$ & 
  open-closed &
  $[3\simpcnt-d,3\simpcnt-b]$ &
  $\Dim-1$ \\ 
Ext &
$b\leq\simpcnt\leq d$ & $\mapsto$ & 
  closed-closed &
  $[b,3\simpcnt-d-1]$ &
  $\Dim$ \\
\midrule
\end{tabular}
\end{center}}
\end{proposition}
\begin{remark}
The types `Ord', `Rel', and `Ext' for intervals in $\Pers_*(\ef)$
are as defined in~\cite{cohen2009extending},
which stand for intervals from the {\it ordinary} sub-barcode,
the {\it relative} sub-barcode,
and the {\it extended} sub-barcode.
\end{remark}

\begin{figure}
\centering
\begin{tikzpicture}[xscale=2.5,yscale=1.5]
\draw (0,0) node(48) {$(L_4,L_8)$} ;
\draw (1,0) node(58) {$(L_5,L_8)$} ;
\draw (2,0) node(68) {$(L_6,L_8)$} ;
\draw (3,0) node(78) {$(L_7,L_8)$} ;
\draw (4,0) node(88) {$(L_8,L_8)$} ;

\draw (0,1) node(47) {$(L_4,L_7)$} ;
\draw (1,1) node(57) {$(L_5,L_7)$} ;
\draw (2,1) node(67) {$(L_6,L_7)$} ;
\draw (3,1) node(77) {$(L_7,L_7)$} ;

\draw (0,2) node(46) {$(L_4,L_6)$} ;
\draw (1,2) node(56) {$(L_5,L_6)$} ;
\draw (2,2) node(66) {$(L_6,L_6)$} ;

\draw (0,3) node(45) {$(L_4,L_5)$} ;
\draw (1,3) node(55) {$(L_5,L_5)$} ;

\draw (0,4) node(44) {$(L_4,L_4)$} ;

\draw[<-] (48) edge node[above]{$\usimp_4$} (58);
\draw[<-] (58) edge node[above]{$\usimp_5$} (68);
\draw[<-] (68) edge node[above]{$\usimp_6$} (78);
\draw[<-] (78) edge node[above]{$\usimp_7$} (88);

\draw[<-] (47) edge node[above]{$\usimp_4$} (57);
\draw[<-] (57) edge node[above]{$\usimp_5$} (67);
\draw[<-] (67) edge node[above]{$\usimp_6$} (77);

\draw[<-] (46) edge node[above]{$\usimp_4$} (56);
\draw[<-] (56) edge node[above]{$\usimp_5$} (66);

\draw[<-] (45) edge node[above]{$\usimp_4$} (55);

\draw[->] (48) edge node[left]{$\usimp_7$} (47);
\draw[->] (47) edge node[left]{$\usimp_6$} (46);
\draw[->] (46) edge node[left]{$\usimp_5$} (45);
\draw[->] (45) edge node[left]{$\usimp_4$} (44);

\draw[->] (58) edge node[left]{$\usimp_7$} (57);
\draw[->] (57) edge node[left]{$\usimp_6$} (56);
\draw[->] (56) edge node[left]{$\usimp_5$} (55);

\draw[->] (68) edge node[left]{$\usimp_7$} (67);
\draw[->] (67) edge node[left]{$\usimp_6$} (66);

\draw[->] (78) edge node[left]{$\usimp_7$} (77);

\draw [decorate,decoration={brace,amplitude=10pt},xshift=-4pt,yshift=0pt]
(-0.25,-0.1) --
node [midway,above,sloped,yshift=10pt] {\rotatebox{-90}{$\ef$}}
(-0.25,4.1);

\draw [decorate,decoration={brace,amplitude=10pt,mirror},xshift=-4pt,yshift=0pt]
(-0.1,-0.3) -- 
node [midway,below,yshift=-10pt] {$\ud$}
(4.4,-0.3);
\end{tikzpicture}
\caption{A Mayer-Vietoris pyramid relating the second half 
of $\ef$ and $\ud$ for $\simpcnt=4$.}
\label{fig:pyramid}
\end{figure}
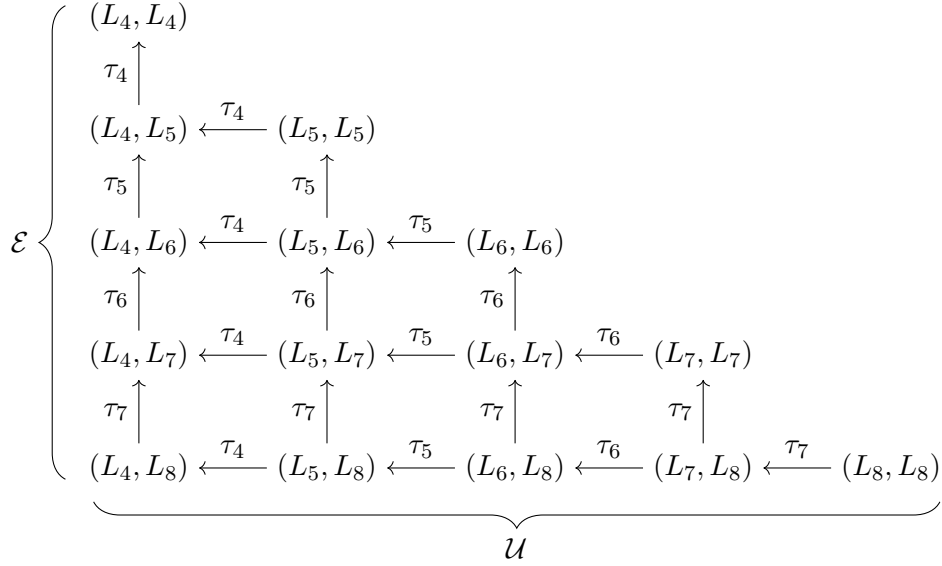

\begin{proof}[Proof $($Sketch$)$]
We can build a Mayer-Vietoris pyramid relating the second half 
of $\ef$ and the second half of $\ud$
similar to the one in~\cite{carlsson2009zigzag-realvalue}.
A pyramid for $\simpcnt=4$ is shown in Figure~\ref{fig:pyramid},
where the second half of $\ef$ is along the left side of the triangle
and the second half of $\ud$ is along the bottom.
In Figure~\ref{fig:pyramid}, we represent the second half of $\ef$ and $\ud$
in a slightly different way
considering that $L_4=K$ and $L_8=\emptyset$.
Also, each vertical arrow indicates
the addition of a simplex in the second complex of the pair
and each horizontal arrow indicates
the deletion of a simplex in the first complex.

To see the correctness of the mapping,
we first note that
each square in the pyramid is a Mayer-Vietoris diamond.
Then, 
the mapping as stated can be verified
using the Diamond Principle (Theorem~\ref{thm:diamond}).
However,
there is a quicker way to verify the mapping in the proposition
by observing the following: 
corresponding intervals in $\Pers_*(\ef)$ and $\Pers_*(\ud)$
have the same set of creator and destroyer if we
ignore whether it is the addition or deletion of a simplex.
For example, an interval in $\Pers_*(\ef)$ may be created by the addition 
of a simplex $\sG$ in the first half of $\ef$ and 
destroyed by the addition of another simplex $\sG'$ in the second half of $\ef$.
Then, its corresponding interval in $\Pers_*(\ud)$ is 
also created by the addition of $\sG$ in the first half but 
destroyed by the {\it deletion} of $\sG'$ in the second half.
Note that the dimension change for the case $b>\simpcnt$ is caused by
the swap of creator and destroyer.
\end{proof}


By Proposition~\ref{prop:ud-f-map} and~\ref{prop:ef-ud-map}, we
only need to compute $\Pers_*(\ef)$
in order to compute $\Pers_*(\Fcal)$.
The barcode of $\ef$ can be computed using the `coning' technique~\cite{cohen2009extending},
which converts $\ef$ into an (absolute) non-zigzag filtration $\hat{\ef}$.
Specifically, let $\oG$ be a vertex different from all vertices in $K$.
The {\it cone} $\oG\cdot\sG$ of a simplex $\sG$ of $K$ 
is the simplex $\sG\union\Set{\oG}$.
The cone $\oG\cdot L_i$ of a complex $L_i$ consists of three parts: 
the vertex $\oG$, $L_i$, and cones of all simplices of $L_i$.
The filtration $\hat{\ef}$ is then defined as~\cite{cohen2009extending}:
\[\hat{\ef}:
\ucplx_0\union \Set{\oG}
\incto
\cdots
\incto
\ucplx_{\simpcnt}\union \Set{\oG}
=K\union\oG\cdot\ucplx_{2\simpcnt}
\incto
K\union\oG\cdot\ucplx_{2\simpcnt-1}
\incto
\cdots
\incto
K\union\oG\cdot\ucplx_{\simpcnt}
\]
We have that $\Pers_*(\ef)$ equals $\Pers_*(\hat{\ef})$
without the only infinite interval~\cite{cohen2009extending}.
Note that if a simplex $\sG$ is added (to the second complex) 
from $(K,\ucplx_{i})$ to $(K,\ucplx_{i-1})$ in $\ef$,
then the cone $\oG\cdot\sG$ is added 
from $K\union\oG\cdot\ucplx_{i}$ to $K\union\oG\cdot\ucplx_{i-1}$ in $\hat{\ef}$.

\bigskip
We now have the conclusion of this section:
\begin{theorem}
The barcode of a non-repetitive, simplex-wise zigzag filtration 
with length $\filtcnt$ can be computed in time $T(\filtcnt)+O(\filtcnt)$,
where $T(\filtcnt)$ is the time used for computing the barcode
of a non-zigzag filtration with length $\filtcnt$.
\end{theorem}
\begin{proof}
Given an input $\Fcal$, we first construct
the filtration $\ud$, the map $\id_\Fcal$, and the filtration $\hat{\ef}$ in linear time.
After running a non-zigzag persistence algorithm on $\hat{\ef}$, 
we can derive $\Pers_*(\Fcal)$ by the mappings in Proposition~\ref{prop:ud-f-map} and~\ref{prop:ef-ud-map},
in linear time.
\end{proof}


\subsection{Experiments}
We implemented our algorithm for non-repetitive filtrations
described in this section and compared the performance
with {\tt Dionysus}~\cite{Dionysus}
and {\tt Dionysus2}~\cite{Dionysus2}, 
which are two versions of implementation of the zigzag persistence 
algorithm (for general filtrations) described in~\cite{carlsson2009zigzag-realvalue}.
In our implementation, we utilize the {\tt Simpers}~\cite{dey2014computing} 
software for the computation of non-zigzag persistence.

To generate non-repetitive filtrations,
we first take a simplicial complex with vertices in $\Real^3$,
and then take the height function $h$ along a certain axis
on the complex.
After this, we build an up-down filtration for the complex 
where the first half is the lower-star filtration of $h$ 
and the second half is the (reversed) upper-star filtration of $h$. 
We then randomly perform outward switches on the up-down filtration
to derive a non-repetitive filtration.
Note that the simplicial complex is derived from a triangular mesh 
with a Vietoris-Rips complex on the vertices as a supplement.

Table~\ref{tab:perform} lists the running time of the 
three algorithms on various non-repetitive filtrations we generate.
The column `Mesh' contains the triangular mesh used to generate
each filtration;
the column `Length' contains the length (i.e., number of additions and deletions) 
of each filtration;
the columns `$\text{T}_\text{ours}$', `$\text{T}_\text{Dio}$', and `$\text{T}_\text{Dio2}$'
contain the running time of our algorithm, {\tt Dionysus}, and {\tt Dionysus2}
respectively;
the column `Speedup' equals $\min(\text{T}_\text{Dio},\text{T}_\text{Dio2})/\text{T}_\text{ours}$.
From the table, the speedup of our algorithm 
is significant: on inputs with about 5 million of additions and deletions,
our algorithm is 50 or 100 times faster than the other two algorithms.
This speedup is expected as our algorithm only needs a slight amount of 
additional processing
other than the non-zigzag persistence computation.
Note that the newer version {\tt Dionysus2} is not always faster
than its older version (i.e., on the filtration of Dragon). Hence,
for a conservative estimate, we took the best
performer of the two versions for comparing with our algorithm.


\newcommand{\tabnum}[1]{\tt\small #1}
\begin{table}[!htb]
\centering
\caption{Running time of 
our algorithm, {\tt Dionysus}, and {\tt Dionysus2} 
on non-repetitive filtrations.
All tests were run on a commodity MacBook Pro
with a 3.1GHz Dual-Core CPU and
8GB memory.}
\label{tab:perform}
\begin{tabular}{lrrrrr}
\midrule
Mesh & \makecell[c]{Length} & \makecell[c]{$\text{T}_\text{ours}$} & \makecell[c]{$\text{T}_\text{Dio}$} & \makecell[c]{$\text{T}_\text{Dio2}$} & Speedup\\
\midrule
Bunny & \tabnum{419,372} & \tabnum{2.9s} & 
  \tabnum{46.6s} & 
  \tabnum{26.4s} &
  \tabnum{9.1}
  \\
Armadillo & \tabnum{2,075,680} & \tabnum{14.8s} & 
  \tabnum{17m27.3s} & 
  \tabnum{8m41.6s} &
  \tabnum{35.2} \\
Hand & \tabnum{5,254,620} & \tabnum{44.6s} & 
  \tabnum{39m38.2s} & 
  \tabnum{23m59.0s} &
  \tabnum{53.3}
  \\
Dragon & \tabnum{5,260,700} & \tabnum{42.6s} & 
  \tabnum{1h11m52.5s} & 
  \tabnum{1h53m33.2s} &
  \tabnum{101.2}
  \\
\midrule
\end{tabular}
\end{table}

\section{Conclusions}
In this paper, we study one of the two types of dualities for persistence~\cite{cohen2009extending,de2011dualities}
 in the zigzag setting, which is the duality for the absolute and relative
zigzag modules. The other duality for persistent
homology and cohomology addressed in~\cite{de2011dualities}
can be adapted to the zigzag setting straightforwardly;
see~\cite[Section 5]{dey2021graph}.
Furthermore, 
the weak duality result presented in Theorem~\ref{thm:dual} 
extends to cohomology modules 
directly
due to the equivalence of barcodes for zigzag persistent homology and cohomology.

Our main finding is
a weak duality result for the absolute and relative
zigzag modules generated from non-repetitive filtrations.
This weak duality led to
two efficient algorithms for non-repetitive filtrations, one for manifolds and
the other in general. Naturally, it raises the question if a similar
result exists for repetitive filtrations.

We have shown that computing
zigzag persistence for non-repetitive filtrations is almost as efficient
as computing standard persistence. 
However, for zigzag filtrations in general, 
the persistence computation~\cite{carlsson2009zigzag-realvalue,maria2014zigzag} 
is still far more costly
than standard persistence (though theoretically the time complexities are the same~\cite{milosavljevic2011zigzag}). This research motivates
the question if further insights into zigzag persistence 
can lead to computing general zigzag persistence 
more efficiently in practice.

\section*{Acknowledgment:}
We thank the Stanford Computer Graphics Laboratory
for providing the triangular meshes used in the experiment
of this paper.


\bibliographystyle{plainurl}
\bibliography{refs}

\appendix

\section{Proof of Theorem~\ref{thm:recover-F-from-KbyF}}
\label{sec:pf-thm-recover-F-from-KbyF}
For the proof, in Table~\ref{tab:dual-rewritten}, we rewrite the correspondence in Theorem~\ref{thm:dual}.
We then have the following:

\begin{table}[!t]
\caption{The correspondence of intervals in Theorem~\ref{thm:dual}
rewritten for the proof of Theorem~\ref{thm:recover-F-from-KbyF}.}
\label{tab:dual-rewritten}
\centering
\begin{tabular}{lllcc}
    \midrule
    \multicolumn{3}{c}{$\Pers(\Hm_{*}(\Fcal))$} & & $\Pers(\Hm_\Dim(K/\Fcal))$ \\ 
    \cmidrule{1-3}\cmidrule{5-5}
    \makecell[c]{Type} & Interval & Dim & & Interval(s)
    \\ 
    \midrule 
    \makecell[l]{closed-open,\\open-closed} & $[\birth,\death]$ & $\Dim-1$ & $\mapsto$ &
    $[\birth,\death]$ 
    \\\cmidrule{1-5} 
    closed-closed & $[\birth,\death]$ & $\Dim$ & $\mapsto$ &
    $[0,\birth-1]$, $[\death+1,\filtcnt]$  
    \\\cmidrule{1-5} 
    open-open & $[\birth,\death]$ & $\Dim-1$ & $\mapsto$ &
    $[0,\death]$, $[\birth,\filtcnt]$  
    \\
    \midrule
\end{tabular}
\end{table}

\begin{itemize}
    \item There is an identity map from closed-open and open-closed 
    intervals in $\Pers(\Hm_{\Dim-1}(\Fcal))$
    to intervals in $\Pers(\Hm_{\Dim}(K/\Fcal))$
    which does not start with $0$ and does not end with $\filtcnt$.
    The reason is that all intervals in $\Pers(\Hm_{\Dim}(K/\Fcal))$
    which are correspondence of closed-closed or open-open intervals
    in $\Pers(\Hm_*(\Fcal))$ either starts with $0$ or ends with $\filtcnt$.
    Therefore, 
    from $\Pers(\Hm_{\Dim}(K/\Fcal))$,
    one can recover closed-open and open-closed 
    intervals in $\Pers(\Hm_{\Dim-1}(\Fcal))$
    in linear time.
    
    \item There is an even number of intervals in $\Pers(\Hm_{\Dim}(K/\Fcal))$
    starting with $0$ or ending with $\filtcnt$. Furthermore,
    these intervals form a pairing such that each pair comes from a closed-closed 
    or open-open interval in $\Pers(\Hm_*(\Fcal))$ as shown in Table~\ref{tab:dual-rewritten}.
    Let $[0,i]$ and $[j,\filtcnt]$ be such a pair.
    By inspecting Table~\ref{tab:dual-rewritten},
    we notice that
    $i$ must be an open death index and $j$ must be an open birth index.
    It follows that the interval $[0,i]$ ends because of the addition of the $\Dim$-simplex $\sG_i$
    and the interval $[j,\filtcnt]$ starts 
    because of the deletion of the $\Dim$-simplex $\sG_{j-1}$.
    We have the following cases: 
    \begin{itemize}
        \item $i<j$: In this case, $[0,i]$ and $[j,\filtcnt]$ are disjoint,
        which means that they must come from the closed-closed interval $[i+1,j-1]\in\Pers(\Hm_\Dim(\Fcal))$ 
        as in Table~\ref{tab:dual-rewritten}.
        Notice that $[i+1,j-1]$ starts because of the addition of $\sG_i$
        and ends because of the deletion of $\sG_{j-1}$.
        Since indeed the zigzag persistence of $\Fcal$ can be independently defined
        on each connected component of $K$,
        we must have that $\sG_i$ and $\sG_{j-1}$
        come from the same connected component of $K$.
        
        \item $i\geq j$: In this case, $[0,i]$ intersects $[j,\filtcnt]$,
        which means that they must come from the open-open interval $[j,i]\in\Pers(\Hm_{\Dim-1}(\Fcal))$.
        Notice that $[j,i]$ starts because of the deletion of $\sG_{j-1}$
        and ends because of the addition of $\sG_i$.
        Then similarly as for the previous case,
        $\sG_i$ and $\sG_{j-1}$ must
        come from the same connected component of $K$.
    \end{itemize}
    
    From the above observations, in order to recover 
    the closed-closed intervals in $\Pers(\Hm_{\Dim}(\Fcal))$
    and the open-open
    intervals in $\Pers(\Hm_{\Dim-1}(\Fcal))$,
    one only needs to do the following,
    which can be done in linear time:
    \begin{itemize}
        \item For each interval $[0,i]\in\Pers(\Hm_{\Dim}(K/\Fcal))$,
        take the $\Dim$-simplex $\sG_i$,
        and for each interval $[j,\filtcnt]\in\Pers(\Hm_{\Dim}(K/\Fcal))$,
        take the $\Dim$-simplex $\sG_{j-1}$.
        Pair all the $\sG_i$'s and $\sG_{j-1}$'s (and hence their corresponding intervals)
        belonging to the same connected component of $K$.
        Note that the pairing is unique because
        closed-closed intervals in $\Pers(\Hm_{\Dim}(\Fcal))$
        and open-open intervals in $\Pers(\Hm_{\Dim-1}(\Fcal))$
        bijectively map to the basis of $\Hm_\Dim(K)$
        and hence bijectively map to the components of $K$
        (this can be seen from the proof of Theorem~\ref{thm:dual}).
        
        \item For each pair $[0,i]$ and $[j,\filtcnt]$ in the previous step,
        if $i<j$, then we have a closed-closed interval $[i+1,j-1]\in\Pers(\Hm_\Dim(\Fcal))$;
        otherwise, we have an open-open interval $[j,i]\in\Pers(\Hm_{\Dim-1}(\Fcal))$.
     \end{itemize}
\end{itemize}

One final thing we need to verify is that 
$\Pers(\Hm_{\Dim}(\Fcal))$ contains only closed-closed intervals,
which follows from the fact that $K$ contains no $(\Dim+1)$-simplices.

\end{document}